\documentclass[11pt]{article}
                                                                                                                                                                                                                                                                                                                                                                                                        \usepackage[margin=1in]{geometry}
                                                                                                                                                                                                                                                                                                                                                                                                            \usepackage{blindtext, rotating}

                                                                                                                                                                                                                                                                                                                                                                                                            \usepackage{graphicx}
                                                                                                                                                                                                                                                                                                                                                                                                            
                                                                                                                                                                                                                                                                                                                                                                                                             \usepackage{bbm}

\usepackage[T1]{fontenc}
\usepackage[utf8]{inputenc}

\usepackage{caption}

\usepackage{enumitem}
                                                                                                                                                                                                                                                                                                                                                                                                                  \usepackage{multirow}
\usepackage[dvipsnames]{xcolor}

                                                                                                                                                                                                                                                                                                                                                                                                      \usepackage{amsthm} 
                                                                                                                                                                                                                                                                                                                                                                                                     \usepackage{framed} 

                                                                                                                                                                                                                                                                                                                                                                                                                                                                                                                                                                                                                                                                                                                                                                                                                               \def \pbb {\mathbb{P}}

\def  \calR {\mathcal{R}}
\def \calS {\mathcal{S}}

\def  \calZ {\mathcal{Z}}

\usepackage{mathtools} 
                                                                                                                                                                                                                                                                                                                                                                                                                                                                                                                                                                                                                                                                                                                                                                                                           \DeclareMathOperator{\In}{In}\DeclareMathOperator{\out}{Out} \DeclareMathOperator{\rank}{rank} \DeclareMathOperator{\Span}{span}

\usepackage{amssymb,amsfonts, amsmath,amscd,dsfont,mathrsfs,pifont}
\usepackage{cite}
\usepackage{hyperref}
\usepackage{xcolor}

\usepackage{tikz}
\usepackage{colortbl}
\usetikzlibrary{tikzmark}

\newtheorem{theorem}{Theorem}
\newtheorem{corollary}{Corollary}
\newtheorem{definition}{Definition}
\newtheorem{lemma}{Lemma}

\newtheorem{remark}{Remark}
\newtheorem{example}{Example}

\begin{document}
\date{}
\title{Private Sequential  Function Computation}
\author{Behrooz~Tahmasebi and Mohammad~Ali~Maddah-Ali 
\thanks{
Behrooz Tahmasebi is with the Department of Electrical Engineering and Computer Science (EECS), Massachusetts  Institute of Technology (MIT), Cambridge, MA, USA  (email: bzt@mit.edu). Mohammad Ali Maddah-Ali is with Nokia Bell Labs, Holmdel, NJ, USA (email: mohammad.maddahali@nokia-bell-labs.com).}
\thanks{
This paper has been presented in part at IEEE ISIT 2019 \cite{private_seq}.}
}

 \renewcommand\footnotemark{}
\maketitle

\begin{abstract}

Consider a system, including a user, $N$ servers,  and $K$ basic functions which are known at all of  the servers. Using the combination of those basic functions, it is possible to construct a wide class of functions. The user wishes to compute a particular combination of the basic functions, by offloading the computation to $N$ servers, while the servers should not obtain any information about which combination of the basic functions is to be computed. The objective is to minimize the total number of queries asked by the user from the servers to achieve the desired result. 

As a first step toward this problem, in this paper, we consider the case where the user is interested in a class  of functions which are composition of the basic functions, while each basic function appears in the composition exactly once. This means that in this case, to ensure privacy, we only require to hide to the order of the basic functions in the desired composition of the user. We further assume that the basic functions are linear and  can be represented by (possibly large scale) matrices. 
We call this problem as \emph{private sequential  function computation}.  We study the capacity $C$, defined as the supremum of the number of desired computations, normalized by the number of computations done at the servers,  subject to the privacy constraint.
We prove that $(1-\frac{1}{N})/ (1-\frac{1}{\max(K,N)}) \le C \le 1$. For the achievability, we show that the user can retrieve the desired order of composition, by choosing a proper order of inquiries among different servers, while keeping the order of computations for each server fixed, irrespective of the desired order of composition. In the end, we develop an information-theoretic converse which results in an  upper bound on the capacity.

\end{abstract}


\allowdisplaybreaks
\section{Introduction}\label{intro}

Outsourcing storage and computation to external parties are the inevitable reaction to the growing size of data and increasing load of processing. 
One of the main challenges in those arrangements is to ensure the privacy of data and algorithms, with minimum overhead in terms of computation, storage,  and communication.

\emph{Private function retrieval} (PFR)~\cite{PFR_jafar, PFR_maddah}, also known as \emph{private computation},  is one recently proposed approach to model the problem of privacy in computation.  In this problem, a  user wants to retrieve a linear combination of  a number of files,  stored  in a number of replicated non-colluding servers,  without revealing any information about the coefficients appeared in the linear combination to the servers \cite{PFR_jafar, PFR_maddah}.  In \cite{PFR_maddah, PFR_jafar}, the authors studied the information theoretic capacity of PFR, defined to be the maximum number of bits of  information about the desired linear computation that can be privately retrieved per bit of downloaded information. In~\cite{PFR_jafar}, the authors fully characterized the capacity of the problem.  This problem is also considered for the coded databases \cite{PFR_linear,PFR_general,PFR_coded}. This is also extended to private computation of arbitrary polynomials on Lagrange coded data~\cite{PFR_polynomial} (see also \cite{obead}), and private inner product retrieval \cite{mousavi}.  

PFR is an extension of the renowned  problem of \emph{private information retrieval} (PIR). In PIR, a user wishes to retrieve a specific file from a database,  duplicated across multiple non-colluding servers, while the file identity must be kept private from the servers \cite{PIR_sudan, PIR_sudan2, PIR_jafar}. In \cite{PIR_jafar}, the basic PIR problem has been investigated from an information-theoretic viewpoint  and its capacity has been characterized. Again, the capacity there is defined as  the maximum number of desired information bits per bit of  download in privacy preserving algorithms. Interestingly, the capacity of PIR and PFR are the same, meaning that there is no extra cost to retrieve  a linear combination of the files, rather than a  pure file. PIR  has been also extended to the   different scenarios, including but not limited to, multiround PIR
\cite{PIR_multiround}, PIR with colluding servers~\cite{PIR_robust}, PIR with coded storage~\cite{PIR_MDScoded,PIR_coded},  PIR with eavesdroppers \cite{PIR_eve}, PIR with adversary
\cite{PIR_adversary}, PIR from wiretap channel
\cite{PIR_wiretap}, cache-aided PIR~\cite{PIR_cach_1,PIR_cach_2,PIR_prefetch}, PIR with coded and colluding servers  \cite{PIR_fereji, PIR_robust_2, PIR_taj}, connections between PIR and distributed storage systems \cite{PIR_dist}, and many other problems 
 \cite{PIR_packet, PIR_uncoded, PIR_trafic, PIR_server, PIR_protocol, PIR_shannon, PIR_graph, PIR_search, sun_delivery, jafar_symm, pooya_multi, jafar_cross, ulukus_noisy, ulukus_mac, jafar_x, jafar_bc, ravi_up, prox, mousavi, kazemi}. 

Another  formulation for private computation is known as  secure multi-party computation (MPC) \cite{MPC_yao, MPC_shamir, MPC_bgw}. In the secure multi-party computation, a group of parties   are trying to perform a computation task on their private inputs without disclosing any information about the inputs to each other \cite{SEC-019}. In other words, the objective in MPC is to keep the inputs secure. In this context, an important question is how to perform such task with the minimum number of servers required. In MPC also, the servers may  collude, up to a given number, in order to gain information about the private inputs. A new formulation of this problem is also  recently proposed, where the communication constraint for the computation of high dimensional inputs is also considered. For more explanations, see \cite{limited_sharing, MPC_maddah, salman, aydin}. 

In this paper, we consider private computation from a different perspective. We introduce the problem of \textit{private function computation} as follows. 
Assume that there are a number of basic and public functions
$\{f_1,f_2,\ldots, f_K\}$. Using the combination of  basic functions,
we can construct a wide class of functions of interest. A user wishes to compute a particular combination 
of  the basic functions for an arbitrary input. The user wants to offload the computation
to $N$ non-colluding servers, each can compute all of the basic functions. However, the user wants to keep the desired
function private. Note that making the function private is equivalent to hide the way we combine the basic functions. 
To achieve the desired result, the user sends a sequence of inquiries to the servers in a recursive manner. The servers answer to the inquiries of the user.  In this paper, it is assumed that the servers are not synchronized. Finally, the user must be able to achieve its desired  result, while the servers should not be able to obtain any information about the desired function, or equivalently the desired combination of the basic functions,  of the user. The objective here is to minimize the number of inquiries sent by the user to the servers under the aforementioned criterion.

As a first step toward studying the problem of private function computation, in this paper, we study a basic version of it, called \textit{private sequential function computation},  in which the class of desired functions which the user is interested in contains only the functions which are composition of the basic functions\footnote{We note that there are applications where the computation of a complicated function is the main goal of the problem, and this function is identical to the composition of a number of basic functions, cf., \cite{optimum}.}. In addition, we assume that each basic function appears in the composition exactly once. Also, the functions are assumed to be linear, and  they can be represented by (possibly large scale) full rank matrices. For example, if the set of basic functions is $\{f_1,f_2,f_3\}$, the desired function by the user is one of these six options $f_1 \circ f_2 \circ f_3$, $f_1 \circ f_3 \circ f_2$, $f_2 \circ f_1 \circ f_3$, $f_2 \circ f_3 \circ f_1$, $f_3 \circ f_1 \circ f_2$ or $ f_3 \circ f_2 \circ f_1$. At the end of computation, the servers should not infer anything about the desired function of the user  among the above 6 options.  Due to this assumption, to ensure the privacy, the only information that is required to be kept private  from the servers is the order of the composition. 
 We further assume that each time the user sends an inquiry, including a vector and the index of one the basic functions, to one of the servers, and waits for that server to return the computation of the corresponding basic function for that vector as the input. For example, the user sends the query $(n, W, k)$ to server $n$, and server $n$ returns $f_k(W)$.   Then the user forms another inquiry for one of the servers, as a function of the initial input vectors, the results of the previous inquiries, and possibly some extra random vectors. In the end, the user should be able to compute the final result, and the servers should not gain any information about the order of composition. 
 
When the number of servers is more than  the number of basic functions, this problem is simple. The user asks each server to compute at most one of the functions. The order of inquires determines the order of the basic functions in the desired function, but each server has no idea about the order.   The challenges start when the number servers is less than number of the basic functions, and thus some servers receive more than one inquiries from the user. The challenges are two folds. The server may infer information about the order of composition in the desired function from
\begin{itemize}
\item the order of functions in different queries that it receives,   
\item the dependency of the inputs in the inquiries that it receives.
\end{itemize}
In this paper, we study capacity of the  private sequential function computation problem, which is defined as the supremum of the number of desired computations per query,  subject to the privacy constraint. We derive lower and upper bounds on the  capacity of the private sequential function computation, denoted by $C$,  as  $(1-\frac{1}{N})/ (1-\frac{1}{\max(K,N)}) \le C \le 1$. We assume that the user has $M$ independent inputs, for some large integer $M$, and the objective is to calculate the desired function for all of those $M$ inputs. We show that the user can compute the desired function for those inputs, by choosing a proper order of inquiries among different servers, while keeping the order of computations for each server \emph{fixed}, irrespective of the desired order of composition in the desired function. Therefore, each server observes a fixed order of computations, and thus gains no information about the desired order of  composition. This will resolve the first challenge raised above. In addition, in each query, the user may add some random vector to the input of the query, to eliminate the dependency of inputs in the queries sent to one server.  These random vectors are reused in queries to other servers in a manner that allows the user the eliminate their contributions in the final results.  This is to resolve the second challenge raised above. 

We then provide an information theoretic converse, which results to an upper bound on the capacity. We note that the inequality $C\le 1$ may seem to be trivial at a glance, but as we show in the paper, this is not the case.

The rest of this paper is organized as follows. In Section 
\ref{prob_state},
we introduce the mathematical model considered in the paper. 
Section
\ref{main_result}
includes the  main result of the paper. 
To prove the theorem, we first provide a number of illustrative examples in Section
 \ref{motivation},
 and then present the achievable scheme in Section
  \ref{achieve}.
  In Section~\ref{proof_achieve}, we prove the feasibility, correctness, and privacy of the proposed scheme. The  proof of the upper bound on the capacity is  presented in Section
   \ref{proof_converse},
and finally, Section
 \ref{conclusion}
 concludes the paper. 
 
\textit{
Notation: For any positive integers $a,b$, let $[a:b]:=\{a,a+1,\ldots, b\}$. We use $\calS_K$ to denote the set of all permutations of $[1:K]$. Using $\Sigma = (\Sigma_K\Sigma_{K-1}\ldots \Sigma_1) \in \calS_K$, we denote the permutation $k \mapsto \Sigma_k$.
Throughout the paper  $\mathbb{F}$ is an arbitrary finite filed. The summation in  $\mathbb{F}$ is denoted by $\oplus$. All the logarithms are in the base of $|\mathbb{F}|$. For two random variables $X,Y$ we write $X\sim Y$ whenever $X$ and $Y$ are identically distributed. The notation $W_{i:j}$ means $(W_i,W_{i+1},\ldots,W_j)$. For two functions $f(n),g(n)$, we write $f(n)= o(g(n))$, when $ \frac{f(n)}{g(n)} \to 0$, as $n \to \infty$. 
}

\section{Problem Statement}\label{prob_state}
Consider a system, including a user, having access to  $M \in \mathbb{N}$ input vectors $W^{(L)}_1,W^{(L)}_2,\ldots,W^{(L)}_M$, chosen independently and  uniformly at random from $\mathbb{F}^L$, for some finite field $\mathbb{F}$ and some large  integer $L$. Thus,
\begin{align}
H(W^{(L)}_1,W^{(L)}_2,\ldots,W^{(L)}_M) = \sum_{m=1}^T H(W_{m}) = ML.
\end{align}
Assume that there is a set of basic functions represented by  $K$ (possibly large scale) matrices $F^{(L)}_1,F^{(L)}_2,\ldots,F^{(L)}_K \in \mathbb{F}^{L \times L}$,  which are distributed independently and uniformly over the set of invertible matrices\footnote{
Note that the set of invertible matrices includes almost all of the matrices, specifically in high dimensions.
} in $\mathbb{F}^{L \times L}$, thus, 
\begin{align}
H(F^{(L)}_1,F^{(L)}_2,\ldots,F^{(L)}_K) = \sum_{k=1}^K H(F^{(L)}_{k}).
\end{align}
For simplicity, throughout the paper, whenever it is not confusing, we drop superscript  $(.)^{(L)}$.

The user wishes to retrieve a specific sequential composition of  the $K$ basic functions of  input vectors.   
The sequence of computations is represented by a permutation $\Sigma  = (\Sigma_K \Sigma_{K-1} \cdots \Sigma_1)$ of $[1:K]$, selected by the user  uniformly at random from the set of all possible permutations. Thus, the user wants the result of  
$F_{\Sigma_K}F_{\Sigma_{K-1}}\cdots F_{\Sigma_2} F_{\Sigma_1} W_{m}$ for all $m \in [1:M]$. This means that the order of composition is fixed for all $M$ input vectors.
The user does not have any information about the matrices\footnote{We make this assumption to prevent the possibility of local computation at the user side. 
 }    $F_k$, $k \in [1:K]$. To obtain the desired result, it relies on $N \in \mathbb{N}$ non-colluding servers, each has access to the full knowledge of $F_k$, $k \in [1:K]$. We also refer to the parameter  $M$ as the number of requests in the paper. 

Each time that a server is called by the user, it receives some  $W \in \mathbb{F}^L$ along with an index $k \in [1:K]$, and returns $F_kW$ to the user. Assume that the servers are not synchronized.
Also, assume that the user utilizes the servers  for $D$ times in total,  in order to achieve its desired result. This means that the user generates  a sequence of queries $Q_1,Q_2,\ldots,Q_D$, where query $Q_d$, $d \in [1:D]$,  is an ordered triple $Q_d = (Q_d^{(\text{server})},  Q_d^{(\text{input})} ,   Q_d^{(\text{function})} )$,  meaning that the user at the $d^{th}$ step, sends $Q_d^{(\text{input})} \in \mathbb{F}^L$ to the server $Q_d^{(\text{server})}\in [1:N]$, and asks it to run the function $Q_d^{(\text{function})} \in [1:K] $. The server then computes the desired result, denoted by $A_d$, $A_d = F_{Q_d^{(\text{function})}}{Q_d^{(\text{input})}}\in \mathbb{F}^L$,  and sends it to the user. 
We refer to $D$ as the number of queries in the paper. 

A $(K,N,M,D,L)$  scheme of private sequential function computation comprises of a sequence of (possibly randomized) encoders $\Phi_d$, $d \in [1:D]$, such that 
$Q_d = \Phi_d(W_{1:M}, A_{1:d-1},\Sigma)$. 
In addition, it includes a sequence of functions (decoders) $\Psi_m$, $m \in [1:M]$,  such that
 $\Psi_m(W_{1:M},A_{1:D},\Sigma)$ is an estimation of $F_{\Sigma_K}F_{\Sigma_{K-1}}\ldots F_{\Sigma_{1}} W_{m}$.

We need to specify a notation for the sequence of queries received by each server. 

\begin{definition}
For any $n\in [1:N]$, the list of queries received by server $n$ is denoted by
\begin{align}
\widetilde{Q}_n := \big (Q_{d} :  d\in [1:D]\  \textrm{and}\  Q^{(\text{server})}_d=n\big ).
\end{align}
We call $\widetilde{Q}_n$ as the marginal list of queries at server $n$. 
\end{definition}

\begin{remark}\normalfont
Note that since servers are not synchronized,  the server $n$ does not have any information about the 
 locations of its received queries in the query list of the user, i.e., the exact values of index $d$ for the  queries are unknown at the server side, although their relative orders are known. 
\end{remark}

We now formally define an achievable scheme. 
\begin{definition}\label{def3}
 For any positive integers $K,N,M,D \in \mathbb{N}$,
 a sequence of $(K,N,M,D,L)$ computation schemes $\{\Phi^{(L)}_{1:D},\Psi^{(L)}_{1:M}\}$, $L=1,2,\ldots$, 
 is said to be $M-$achievable, iff  
 \begin{itemize}
\item $[$Correctness$]$ for any $m \in [1:M]$,
\begin{align}
\mathbb{P} \Big (\Psi^{(L)}_{m}(W^{(L)}_{1:M},A^{(L)}_{1:D},\Sigma) \neq F_{\Sigma_K}^{(L)}F_{\Sigma_{K-1}}^{(L)} \cdots F_{\Sigma_{1}}^{(L)} W_{m}^{(L)}\Big )\to 0,
\end{align}
as $L \to \infty$. 

\item $[$Privacy$]$ for any $n \in [1:N],L\in \mathbb{N}$, the marginal list of the  $n^{th}$ server $\widetilde{Q}^{(L)}_n$   must be independent of the order of composition, i.e., $I(\widetilde{Q}^{(L)}_n,F^{(L)}_{1:K};\Sigma) =0$.
\end{itemize}
 \end{definition}

Now we present the definitions of an achievable rate and the capacity of the private sequential function  computation problem. 
 
 \begin{definition}
The rate of a  $(K,N,M,D,L)$ computation scheme is defined as $R=\frac{KM}{D}$. A rate $R$ is said to be $M-$achievable, iff there exists a sequence of $M-$achievable  computation schemes  $(K,N,M,D,L)$, $L=1,2,\ldots$, such that their rate is least $R$. 
 \end{definition}

\begin{remark}\normalfont
The motivation of this definition for the rate is that in  a  $(K,N,M,D,L)$ computation scheme,  the user wants to compute a composition of $K$ functions on $M$ input vectors (total of $KM$ computations)  and it utilizes the servers for $D$ times.
\end{remark}

\begin{definition}
A  positive rate $R$ is said to be achievable, if there is a sequence of $M-$achievable rates $\{R_M\}_{M\in \mathbb{N}}$, such that $\lim\limits_{M  \to \infty} R_M=R$.
\end{definition}

\begin{definition}
The capacity of the private sequential  function computation, denoted by $C$, is defined as  
the supremum of all achievable rates.
\end{definition}

\section{Main Results}\label{main_result}

The main result of this paper is summarized in the following theorem.

\begin{theorem}\label{thrm1}
The capacity of the private sequential function computation problem satisfies the following inequality:
\begin{align}
({1-\frac{1}{N}})/({1-\frac{1}{\max(K,N)}}) \le C \le 1,
\end{align}
where $K$ is the number of basic functions to be computed in a sequential and private manner, and $N$ is the number of non-colluding servers. 
\end{theorem}

The general achievable scheme for Theorem \ref{thrm1} can be found in Section \ref{achieve}, the formal proofs of correctness and privacy of the proposed algorithm are presented  in Section  \ref{proof_achieve}, and the converse proof of Theorem \ref{thrm1} is detailed in Section \ref{proof_converse}. 

\begin{remark}\normalfont
When $K \leq N$, the capacity  is equal to one. This means that if the number of functions does not exceed the number of servers, then one can achieve the private computation without  any extra cost.  This should not be surprising, because a simple 
achievable scheme where the computation of each basic function is assigned to one specific server
 ensures privacy (see Example \ref{exmpl_1}). 
\end{remark}

\begin{remark}\normalfont
 When $K>N$, the user wishes to compute a composition of a  number of basic functions which is greater than the number of available servers. Intuitively speaking, in this case there is at least one server that must compute at least two basic functions. This means that to achieve privacy, the order of computations at that  server  should not leak information about the desired order of composition. To ensure privacy in this case, we propose a scheme which has a surprising feature:  the order of computations at each server is fixed, irrespective of the desired order of computations.  The user can retrieve any order of composition, by selecting an appropriate order for queries $Q_1,Q_2,\ldots,Q_D$, such that the order of computations  at each server remains fixed. In addition, some randomness is added to the inputs, such that the sever cannot infer any order from dependencies between inputs of the different queries.  The proposed scheme satisfies both privacy and correctness with zero probability of error, and the achievable  rate of $({1-\frac{1}{N}})/({1-\frac{1}{K}})$  as $M\to \infty$. 
\end{remark}

\begin{remark}\normalfont 
Let us explain the differences between the problem of private sequential function computation and the other formulations for the problem of private computation. 
\begin{itemize}
\item In the PFR problem~\cite{PFR_jafar, PFR_maddah}, there are $K$ vectors $W_k$, $k \in [1:K]$, available at $N$ non-colluding replicated servers,  and a user wishes to retrieve the result of $\sum_{k=1}^K\alpha_k W_k$, for some scalars $\alpha_k$, $k \in [1:K]$. The goal there is to keep the coefficient $\alpha_k$, $k \in [1:K]$ in the linear combination private. Let $F:= [W1,W_2,\cdots,W_K]\in \mathbb{F}^{L\times K}$ and $X = [\alpha_1,\alpha_2,\cdots,\alpha_K]^T$. Therefore, by the new notations, in the PFR problem, the user wishes to retrieve $FX$, while it is required to keep the vector $X$ private from the servers. However, in the problem of private sequential function computation, the user want to retrieve, e.g., $F_{\Sigma_K}F_{\Sigma_{K-1}}\cdots F_{\Sigma_2} F_{\Sigma_1} X$, for some permutation $ (\Sigma_K \Sigma_{K-1} \cdots \Sigma_1)$ of $[1:K]$  while  the privacy means the the servers should not obtain any information about the permutation $ (\Sigma_K \Sigma_{K-1} \cdots \Sigma_1)$.

\item In the secure MPC model, each one of $K$ parties has a private input $F_k$, $k\in [1:K]$, and the final goal is to compute $p(F_1,F_2,\cdots,F_K)\times X$, where $p(.)$ is a polynomial and  $X$ is a vector, using $N$ servers. In MPC,  the privacy condition requires that in the procedure of computation, the servers and other parties learn nothing about the private inputs $F_k$, $k\in [1:K]$. However, in the problem of  private sequential function computation, all servers are fully aware of $F_1,F_2,\cdots,F_K$, and the objective is to keep the order of composition $ (\Sigma_K \Sigma_{K-1} \cdots \Sigma_1)$ private. 
\end{itemize}
\end{remark}

 \section{Motivation and Examples}
\label{motivation}

To motivate the achievable scheme, in this section, we explain the proposed scheme through some examples.

\begin{example}[$K=2$ and  $N=2$]\label{exmpl_1} \normalfont
This examples explain why for $K\le N$, Theorem \ref{thrm1} states that the capacity of the private sequential function computation is equal to one.
In this example, we have $N=2$ servers and $K=2$ basic functions. Also, assume that the user has only $M=1$ input vector $W_1$. Note that in this case we have only two distinct orders of compositions; $\sigma = (1~2)$ and $\sigma=(2~1)$. One simple achievable scheme is as follows:
\begin{itemize}
\item $F_2(F_1W_1).$  In this case,  the user asks  the first server to compute the function $F_1$ on  $W_1$, and then receives $F_1W_1$. Then, the user asks the second server to compute the function $F_2$ on $F_1W_1$. We use the following schematic to demonstrate this   scheme:
\begin{eqnarray*}
\begin{array}{|cc|}
\multicolumn{2}{c}{\sigma = (2~1)}\\ \hline
\mbox{\tiny SERVER $1$}&\mbox{\tiny SERVER $2$}\\\hline
F_1W_1~\tikzmark{1x1}&\tikzmark{1y1}~F_2(F_1W_1)\\\hline
\end{array}
\end{eqnarray*}

\item $F_1(F_2W_1).$ In this case, first the user asks the second server to  compute $F_2$ on $W_1$, and  it  receives $F_2W_1$. Then, the user asks  the first server to compute $F_1$ on $F_2W_1$, see the following:
\begin{eqnarray*}
\begin{array}{|cc|}
\multicolumn{2}{c}{\sigma = (1~2)}\\ \hline
\mbox{\tiny SERVER $1$}&\mbox{\tiny SERVER $2$}\\\hline
F_1(F_2W_1) ~\tikzmark{2y1}&\tikzmark{2x1}~F_2W_1 \\\hline
\end{array}
\end{eqnarray*}

\end{itemize}

Trivially, the proposed scheme satisfies the correctness property.  It is also private. Note that in both cases the first server receives a vector and computes $F_1$, and the second server receives a vector, and computes $F_2$. There is no way for each server to distinguish the two cases. Recall that $F_1,F_2$ are full rank matrices, and since $W_1$ is distributed uniformly over $\mathbb{F}^L$, $F_1W_1$  and $F_2W_1$ are also  distributed uniformly over $\mathbb{F}^L$.

\end{example}

Generally, when $K \le N$, one can employ a similar approach to 
achieve the capacity in a one-shot zero-error setting as follows:
\begin{eqnarray*}
\begin{array}{|cccc|}
\multicolumn{4}{c}{\sigma = (\sigma_K~ \sigma_{K-1} \cdots ~ \sigma_1)}\\ \hline
\mbox{\tiny SERVER $\sigma_1$}&\mbox{\tiny SERVER $\sigma_2$}&\mbox{$\cdots$}&\mbox{\tiny SERVER $\sigma_K$}\\\hline
F_{\sigma_1}W_1~\tikzmark{a11}&\tikzmark{b11}~F_{\sigma_2}(F_{\sigma_1}W_1)~\tikzmark{c11}&\tikzmark{d11}~\ldots~\tikzmark{e11}&\tikzmark{f11}~F_{\sigma_K}(\cdots (F_{\sigma_1}W_1)\cdots)\\\hline
\end{array}
\end{eqnarray*}
\begin{tikzpicture}[overlay, remember picture]
    \draw [black, thick, ->, rounded corners=8pt] ({pic cs:a11})  ->  ({pic cs:b11});
      \draw [black, thick, ->, rounded corners=8pt] ({pic cs:c11})  ->  ({pic cs:d11});
        \draw [black, thick, ->, rounded corners=8pt] ({pic cs:e11})  ->  ({pic cs:f11});
     \end{tikzpicture}

In the above scheme, the arrows demonstrate the order of asking the queries by the user.  Accordingly, to compute $F_{\sigma_K}F_{\sigma_{K-1}}\cdots F_{\sigma_2} F_{\sigma_1} W_{1}$, the user first asks server ${\sigma_1}$ to compute $F_{\sigma_1} W_{1}$ with  $W_{1}$ as the input, then it asks server 
 ${\sigma_2}$ to compute $F_{\sigma_2}F_{\sigma_1} W_{1}$ with  $F_{\sigma_1} W_{1}$, and so on. 
 
Thus, irrespective of the order $ (\sigma_K \sigma_{K-1} \cdots \sigma_1)$,  the user always asks server $n$, $n \in [1:K]$, to compute function $F_{n}$.  Servers $K+1$ to $N$ remain idle. Thus privacy is preserved.

\begin{example}[$K=3$ and  $N=2$] \label{exmpl_2} \normalfont
This example demonstrates how to achieve the lower bound on the capacity for cases $K > N$.  Consider the problem of private sequential  function  computation with $N=2$ servers and $K=3$ functions. Assume that the user has $M$ input vectors $W_{1:M}$,  for some integer $M$,  and wants to derive $F_{\sigma_3}(F_{\sigma_2}(F_{\sigma_1}W_{1:M}))$, i.e., the desired order is $\sigma=(\sigma_3~\sigma_2~ \sigma_1)$. To explain the challenges of designing an achievable scheme for this case, let us try some solutions. 

\begin{itemize}
\item ($1^{st}$ try)  Consider the following computation scheme, which inspired from what we proposed in the previous example:
\begin{eqnarray*}
\begin{array}{|cc|}
\multicolumn{2}{c}{\sigma = (3~2~1)}\\ \hline
\mbox{\tiny SERVER $1$}&\mbox{\tiny SERVER $2$}\\\hline
F_1W_1&F_2(F_1W_1)\\\hline
F_3(F_2(F_1W_1))& -\\\hline
\end{array}
\end{eqnarray*}
From the above table, one can easy infer the order of queries are as follows: $(\textrm{Server} \ 1 ,  W_1, F_1)$,  $(\textrm{Server}\ 2, F_1W_1,  F_2)$, $(\textrm{Server}\ 1, F_2(F_1W_1), F_3)$, where the input of each query is the output of the previous query.   In other words,  first server one computes $F_1W_1$ with input $W_1$, then server 2 computes $F_2(F_1W_1)$ with input $F_1W_1$, then server one computes $F_3(F_2(F_1W_1))$ with input $F_2(F_1W_1)$. 
This computation scheme breaches information about  $\sigma$. The reason is that server one in the first query that it receives,  computes produces $F_1W_1$, and in the second query, it receives $F_2(F_1W_1))$ as input. With a little effort, it realizes that the second input is indeed equal to $F_2$ times the first output. Thus server one realizes that the desired order of composition is indeed $F_1$, then $F_2$, and finally $F_3$.

\item  ($2^{nd}$ try)  In the above scheme, one server can apply reverse computation and gain information about the order of computations. One possible solution  is to add randomness to the queries  to ensure the privacy; see the following:
\begin{eqnarray*}
\begin{array}{|cc|}
\multicolumn{2}{c}{\sigma = (3~2~1)}\\ \hline
\mbox{\tiny SERVER $1$}&\mbox{\tiny SERVER $2$}\\\hline
F_1W_1&F_2(F_1W_1)\\\hline
F_3(F_2(F_1W_1)\oplus {\color{red} Z})& F_3{\color{red} Z}\\\hline
\end{array}
\end{eqnarray*}
In the above scheme, the user chooses $Z$ uniformly at random from $\mathbb{F}^L$. From the above table, one can easy infer the order of queries as $(\textrm{Server} \ 1,  W_1, F_1)$, $(\textrm{Server} \ 2,  F_1W_1, F_2)$, $(\textrm{Server}\ 1, F_2(F_1W_1)\oplus { Z}, F_3)$, and $(\textrm{Server}\ 2,  {Z}, F_2)$. Note that those quires can be formed by the user from the result of  previous queries, and without any knowledge of $F_k$, $k=1,2,3$.  In addition, it is easy to verify that the user can calculate $F_3(F_2(F_1W_1))$ by subtracting  the result of the last query $F_3\oplus{ Z}$ from the result of the third query $F_3(F_2(F_1W_1)\oplus {Z})$. 

Is the above scheme private? The positive aspect of the above scheme is that with reverse computation, no information can be inferred by the servers. More precisely, the inputs of the two queries of server one, i.e., $W_1$ and $F_2(F_1W_1)\oplus {Z}$,  are independent, given the knowledge of $F_k$, $k=1,2,3$. Similarly the inputs of the two queries of server 2,   i.e.,  $F_1W_1$ and  ${ Z}$,   are also independent, given the knowledge of $F_k$, $k=1,2,3$. However, it is not enough to guarantee privacy. The reason is that, in the above scheme, for the other orders of composition $\sigma$, it is required to ask different orders of inquiries from each servers. For example,  if the desired order is $\sigma = (2~1~3)$, then the order of computation at server one is $F_3 \to F_2$, i.e.,  server 1 first computes  $F_3$ and then $F_2$, and the order of computation at server two is $F_1 \to F_2$, i.e.,  server 2 first computes  $F_1$ and then $F_2$. In other words, the order of computation at each server depends on $\sigma $, which is the violation of privacy constraint.  

\item (Complete solution) From the first two tries, we observe that there are two challenges in designing  the achievable scheme, (1)  dependency of the inputs of the different queries to each server, which allows  information leakage by reverse computation, and (2) dependency of the order of computation at each server to the desired order of composition   $\sigma$. To solve these issues, we introduce the following computation scheme. The proposed scheme eliminates the possibility of reverse computation by adding randomness.  
It also uses fixed order of computation at each server irrespective to $\sigma$. The proposed solutions for all six possible choices of $\sigma$ are illustrated in the following tables. Note that the arrows show the order of asking the queries and 
 the random vectors $Z_{1:M+2}$ are distributed independently and uniformly over $\mathbb{F}^L$. 
\end{itemize}
\newpage

\begin{eqnarray*}
\begin{array}{|c|ll|} 
\multicolumn{3}{c}{\sigma=(3~2~1)} \\ \hline
\mbox{\tiny{BLOCK}}&\mbox{\tiny ~~~~~~SERVER $1$}&\mbox{\tiny ~~~~~~~~SERVER $2$}\\\hline
1&
\begin{array}{l}
 F_1W_1~\tikzmark{x21}\\
 F_3{\color{red} Z_1}~\tikzmark{z21}
\end{array}
&
\begin{array}{l}
\tikzmark{y21}~F_2(F_1W_1) \\
\tikzmark{r21}~ F_3(F_2(F_1W_1) \oplus {\color{red} Z_1})
\end{array}\\ \hline
2&
\begin{array}{l}
 F_1W_2~\tikzmark{x22}\\
 F_3{\color{red} Z_2}~\tikzmark{z22}
\end{array}
&
\begin{array}{l}
\tikzmark{y22}~F_2(F_1W_2) \\
\tikzmark{r22}~ F_3(F_2(F_1W_2) \oplus {\color{red} Z_2})
\end{array}\\ \hline
$\vdots$&
$~~~~~\vdots$
&
$~~~~~\vdots$\\ \hline
M&
\begin{array}{l}
 F_1W_M~\tikzmark{x2T}\\
 F_3{\color{red} Z_M}~~~~\tikzmark{z2T}
\end{array}
&
\begin{array}{l}
\tikzmark{y2T}~F_2(F_1W_M) \\
\tikzmark{r2T}~ F_3(F_2(F_1W_M) \oplus {\color{red} Z_T})
\end{array}\\ \hline
M+1&
\begin{array}{l}
 F_1{\color{red} Z_{M+1}}~\tikzmark{x2T1}
\end{array}
&
\begin{array}{l}
\tikzmark{y2T1}~F_2{\color{red} Z_{M+2}}
\end{array}\\ \hline
\end{array}
~
\begin{array}{|c|ll|} 
\multicolumn{3}{c}{\sigma=(3~1~2)} \\ \hline
\mbox{\tiny{BLOCK}}&\mbox{\tiny ~~~~SERVER $1$}&\mbox{\tiny SERVER $2$}\\\hline
1&
\begin{array}{l}
 F_1(F_2W_1)~~~~~~~~~~~~\tikzmark{x212}\\
F_3(F_1(F_2W_1) \oplus {\color{red} Z_1}) ~\tikzmark{z212}
\end{array}
&
\begin{array}{l}
\tikzmark{y212}~F_2W_1 \\
\tikzmark{r212}~F_3{\color{red} Z_1}
\end{array}\\ \hline
2&
\begin{array}{l}
 F_1(F_2W_2)~~~~~~~~~~~~~\tikzmark{x222}\\
F_3(F_1(F_2W_2)\oplus {\color{red} Z_2}) ~\tikzmark{z222}
\end{array}
&
\begin{array}{l}
\tikzmark{y222}~F_2W_2 \\
\tikzmark{r222}~ F_3{\color{red} Z_2}
\end{array}\\ \hline
$\vdots$&

$~~~~~\vdots$
&
$~~~~~\vdots$\\ \hline
M&
\begin{array}{l}
 F_1(F_2W_M)~~~~~~~~~~~~~\tikzmark{x2T2}\\
F_3(F_1(F_2W_M) \oplus {\color{red} Z_T})~\tikzmark{z2T2}
\end{array}
&
\begin{array}{l}
\tikzmark{y2T2}~F_2W_M \\
\tikzmark{r2T2}~  F_3{\color{red} Z_M}
\end{array}\\ \hline
M+1&
\begin{array}{l}
 F_1{\color{red} Z_{M+2}}~\tikzmark{x2T11}
 \end{array}
&
\begin{array}{l}
\tikzmark{y2T12}~F_2{\color{red} Z_{M+1}}
\end{array}\\ \hline
\end{array}
\end{eqnarray*}
\begin{tikzpicture}[overlay, remember picture]
    \draw [black, thick, rounded corners=8pt] ({pic cs:x21})  --  ({pic cs:y21})  -- ({pic cs:r21})  -- ({pic cs:z21}) --
   ({pic cs:x22})  --  ({pic cs:y22})  -- ({pic cs:r22}) -- ({pic cs:z22})   ;
     \draw [black, thick, ->, rounded corners=8pt] ({pic cs:x2T})  ->  ({pic cs:y2T})  ->    ({pic cs:r2T})  -> ({pic cs:z2T})  ->
      ({pic cs:x2T1})  ->  ({pic cs:y2T1}) ;
      \draw [black, thick,rounded corners=8pt]   ({pic cs:y212})  -- ({pic cs:x212})  -- ({pic cs:z212})  -- ({pic cs:r212})  --
   ({pic cs:y222}) -- ({pic cs:x222})   -- ({pic cs:z222}) -- ({pic cs:r222}) ;
     \draw [black, thick, ->, rounded corners=8pt]  ({pic cs:y2T2})  ->  ({pic cs:x2T2})   ->  ({pic cs:z2T2}) -> ({pic cs:r2T2}) ->
      ({pic cs:y2T12}) -> ({pic cs:x2T11})  ;
     \end{tikzpicture}

\begin{eqnarray*}
\begin{array}{|c|ll|} 
\multicolumn{3}{c}{\sigma=(2~3~1)} \\ \hline
\mbox{\tiny{BLOCK}}&\mbox{\tiny ~~~~~~SERVER $1$}&\mbox{\tiny ~~~~~~SERVER $2$}\\\hline
1&
\begin{array}{l}
 F_1W_1~\tikzmark{a21}\\
F_3{\color{red} Z_2}~\tikzmark{c21}
\end{array}
&
\begin{array}{l}
 \tikzmark{b21}~F_2{\color{red} Z_1} \\
\tikzmark{d21}~  F_3(F_1W_1\oplus {\color{red} Z_2})
\end{array}\\ \hline
2&
\begin{array}{l}
 F_1W_2~\tikzmark{a22}\\
F_3{\color{red} Z_3}~\tikzmark{c22}
\end{array}
&
\begin{array}{l}
\tikzmark{b22}~F_2(F_3(F_1W_1)) \\
\tikzmark{d22}~  F_3(F_1W_2 \oplus {\color{red} Z_3})
\end{array}\\ \hline
3&
\begin{array}{l}
\cdots \\
\cdots
\end{array}
&
\begin{array}{l}
~\tikzmark{b23} ~F_2(F_3(F_1W_2)) \\
\cdots
\end{array}\\ \hline
$\vdots$&
$~~~~~\vdots$
&
 $~~~~~\vdots$   \\ \hline
M&
\begin{array}{l}
 F_1W_M~\tikzmark{a2T}\\
F_3{\color{red} Z_{M+1}}~\tikzmark{c2T}
\end{array}
&
\begin{array}{l}
\tikzmark{b2T}~\cdots \\
\tikzmark{d2T}~  F_3(F_1W_M \oplus {\color{red} Z_{M+1}})
\end{array}\\ \hline
M+1&
\begin{array}{l}
 F_1{\color{red} Z_{M+2}}~\tikzmark{a2T1}
\end{array}
&
\begin{array}{l}
\tikzmark{b2T1}~F_2(F_3(F_1W_M))
\end{array}\\ \hline
\end{array}
~
\begin{array}{|c|ll|} 
\multicolumn{3}{c}{\sigma=(2~1~3)} \\ \hline
\mbox{\tiny{BLOCK}}&\mbox{\tiny ~~~~~~SERVER $1$}&\mbox{\tiny ~~~~~~SERVER $2$}\\\hline
1&
\begin{array}{l}
 F_1{\color{red} Z_1}~\tikzmark{a212}\\
F_3(W_1 \oplus {\color{red} Z_3})~\tikzmark{c212}
\end{array}
&
\begin{array}{l}
 \tikzmark{b212}~F_2{\color{red} Z_2} \\
\tikzmark{d212}~  F_3{\color{red} Z_3}
\end{array}\\ \hline
2&
\begin{array}{l}
 F_1(F_3W_1)~~~~~\tikzmark{a222}\\
F_3(W_2 \oplus {\color{red} Z_4})~\tikzmark{c222}
\end{array}
&
\begin{array}{l}
\tikzmark{b222}~F_2(F_1(F_3W_1)) \\
\tikzmark{d222}~  F_3{\color{red} Z_4}
\end{array}\\ \hline
3&
\begin{array}{l}
 F_1(F_3W_2)~~~~~\tikzmark{a232} \\
\cdots
\end{array}
&
\begin{array}{l}
~\tikzmark{b232} ~F_2(F_1(F_3W_2)) \\
\cdots
\end{array}\\ \hline
$\vdots$&
$~~~~~\vdots$
&
 $~~~~~\vdots$   \\ \hline
M&
\begin{array}{l}
\cdots~\tikzmark{a2T2}\\
F_3(W_M\oplus {\color{red} Z_{M+2}})~\tikzmark{c2T2}
\end{array}
&
\begin{array}{l}
\tikzmark{b2T2}~\cdots \\
\tikzmark{d2T2}~  F_3{\color{red} Z_{M+2}}
\end{array}\\ \hline
M+1&
\begin{array}{l}
 F_1(F_3W_M)~~~~~~~~\tikzmark{a2T12}
\end{array}
&
\begin{array}{l}
\tikzmark{b2T12}~F_2(F_1(F_3W_M))
\end{array}\\ \hline
\end{array}
\end{eqnarray*}
\begin{tikzpicture}[overlay, remember picture]
    \draw [black, thick,  rounded corners=8pt]({pic cs:b21})  ->   ({pic cs:a21})  ->   ({pic cs:c21}) ->  ({pic cs:d21}) ->  ({pic cs:b22})   ->
     ({pic cs:a22})  ->      ({pic cs:c22})   -> ({pic cs:d22}) -> ({pic cs:b23})  ;
     \draw [black, thick, ->, rounded corners=8pt] ({pic cs:a2T})  ->  ({pic cs:c2T}) ->  ({pic cs:d2T}) ->   ({pic cs:b2T1}) ->  ({pic cs:a2T1}) ;
    \draw [black, thick,  rounded corners=8pt] ({pic cs:a212})  -- ({pic cs:b212})  --   ({pic cs:d212}) --    ({pic cs:c212}) --   ({pic cs:a222})  -- ({pic cs:b222})   --
          ({pic cs:d222})   -- ({pic cs:c222}) -- ({pic cs:a232}) --  ({pic cs:b232})  ;
     \draw [black, thick,->,  rounded corners=8pt]   ({pic cs:d2T2}) -- ({pic cs:c2T2}) --    ({pic cs:a2T12}) ->  ({pic cs:b2T12}) ;
     \end{tikzpicture}


  \begin{eqnarray*}
  \begin{array}{|c|ll|} 
\multicolumn{3}{c}{\sigma=(1~3~2)} \\ \hline
\mbox{\tiny{BLOCK}}&
\mbox{\tiny ~~~~~~SERVER $1$}
&\mbox{\tiny SERVER $2$}\\\hline
1&
\begin{array}{l}
 F_1{\color{red} Z_1}~\tikzmark{a312}\\
F_3(F_2W_1 \oplus {\color{red} Z_2})~\tikzmark{c312}
\end{array}
&
\begin{array}{l}
 \tikzmark{b312}~F_2W_1 \\
\tikzmark{d312}~  F_3{\color{red} Z_2}
\end{array}\\ \hline
2&
\begin{array}{l}
 F_1(F_3(F_2W_1))~\tikzmark{a322}\\
F_3(F_2W_2 \oplus {\color{red} Z_3})~\tikzmark{c322}
\end{array}
&
\begin{array}{l}
\tikzmark{b322}~F_2W_2 \\
\tikzmark{d322}~  F_3{\color{red} Z_3}
\end{array}\\ \hline
3&
\begin{array}{l}
 F_1(F_3(F_2W_2))~\tikzmark{a332} \\
\cdots
\end{array}
&
\begin{array}{l}
\cdots \\
\cdots
\end{array}\\ \hline
$\vdots$&
$~~~~~\vdots$
&
 $~~~~~\vdots$   \\ \hline
M&
\begin{array}{l}
\cdots~\tikzmark{a3T2}\\
F_3(F_2W_M \oplus {\color{red} Z_{M+1}})~\tikzmark{c3T2}
\end{array}
&
\begin{array}{l}
\tikzmark{b3T2}~F_2W_M\\
\tikzmark{d3T2}~  F_3{\color{red} Z_{M+1}}
\end{array}\\ \hline
M+1&
\begin{array}{l}
 F_1(F_3F_2(W_M))~~~~~\tikzmark{a3T12}
\end{array}
&
\begin{array}{l}
\tikzmark{b3T12}~F_2{\color{red} Z_{M+2}}
\end{array}\\ \hline
\end{array}
~
\begin{array}{|c|ll|} 
\multicolumn{3}{c}{\sigma=(1~2~3)} \\ \hline
\mbox{\tiny{BLOCK}}&\mbox{\tiny ~~~~~~SERVER $1$}&\mbox{\tiny~~~~~~ SERVER $2$}\\\hline
1&
\begin{array}{l}
 F_1{\color{red} Z_2}~\tikzmark{a31}\\
F_3{\color{red} Z_3}~\tikzmark{c31}
\end{array}
&
\begin{array}{l}
 \tikzmark{b31}~F_2{\color{red} Z_1} \\
\tikzmark{d31}~  F_3(W_1\oplus {\color{red} Z_3})
\end{array}\\ \hline
2&
\begin{array}{l}
 F_1(F_2(F_3W_1))~\tikzmark{a32}\\
F_3{\color{red} Z_4}~~~~~~~~~~~~\tikzmark{c32}
\end{array}
&
\begin{array}{l}
\tikzmark{b32}~F_2(F_3W_1) \\
\tikzmark{d32}~  F_3(W_2 \oplus {\color{red} Z_4})
\end{array}\\ \hline
3&
\begin{array}{l}
F_1(F_2(F_3W_2))~\tikzmark{a33}  \\
\cdots
\end{array}
&
\begin{array}{l}
~\tikzmark{b33} ~F_2(F_3W_2) \\
\cdots
\end{array}\\ \hline
$\vdots$&
$~~~~~\vdots$
&
 $~~~~~\vdots$   \\ \hline
M&
\begin{array}{l}
\cdots~\tikzmark{a3T}\\
F_3{\color{red} Z_{M+2}}~\tikzmark{c3T}
\end{array}
&
\begin{array}{l}
\tikzmark{b3T}~\cdots \\
\tikzmark{d3T}~  F_3(W_M \oplus {\color{red} Z_{M+2}})
\end{array}\\ \hline
M+1&
\begin{array}{l}
 F_1(F_2(F_3W_M))~\tikzmark{a3T1}
\end{array}
&
\begin{array}{l}
\tikzmark{b3T1}~F_2(F_3W_M)
\end{array}\\ \hline
\end{array}
\end{eqnarray*}
\begin{tikzpicture}[overlay, remember picture]
    \draw [black, thick,  rounded corners=8pt]({pic cs:b31})  ->   ({pic cs:a31})  ->   ({pic cs:c31}) ->  ({pic cs:d31}) ->  ({pic cs:b32})   ->
     ({pic cs:a32})  ->      ({pic cs:c32})   -> ({pic cs:d32}) -> ({pic cs:b33})  -> ({pic cs:a33})  ;
     \draw [black, thick, ->, rounded corners=8pt]  ({pic cs:c3T}) ->  ({pic cs:d3T}) ->   ({pic cs:b3T1}) ->  ({pic cs:a3T1}) ;
    \draw [black, thick,  rounded corners=8pt] ({pic cs:a312})  -- ({pic cs:b312})  --   ({pic cs:d312}) --    ({pic cs:c312}) --   ({pic cs:a322})  -- ({pic cs:b322})   --
          ({pic cs:d322})   -- ({pic cs:c322}) -- ({pic cs:a332}) ;
     \draw [black, thick,->,  rounded corners=8pt]   ({pic cs:b3T2})  -> ({pic cs:d3T2}) -- ({pic cs:c3T2}) --    ({pic cs:a3T12}) ->  ({pic cs:b3T12}) ;
     \end{tikzpicture}

Let us describe the solution for the case $\sigma=(3~2~1)$ as an example. From the table, one can easily infer the order of queries is $(\textrm{Server}\ 1,  W_1, F_1) \to (\textrm{Server}\ 2, F_1W_1, F_2) \to (\textrm{Server}\ 1, Z_1, F_3) \to (\textrm{Server}\ 2,F_2(F_1W_1) \oplus { Z_1},  F_3)  \to (\textrm{Server}\ 1, W_2,  F_1) \to (\textrm{Server}\ 2,  F_1W_2, F_2) \to (\textrm{Server}\ 1, Z_2, F_3) \to (\textrm{Server}\ 2,  F_2(F_1W_2) \oplus {Z_2}, F_3) \to \ldots$. We note that the input of each query is a function of the outputs of the previous queries.  In addition, the user can form these queries without any  knowledge about the basic functions, and without any matrix multiplication. 
In addition, in the end, the user can calculate $F_3(F_2(F_1W_m))$ by subtracting $F_3 Z_m$ from $F_3( F_2(F_1W_m) \oplus Z_m)$, for $m=1, \ldots M$,  thus correctness is verified. Similarly, the correctness for all six order of desired computation can be verified. 

It remains to establish privacy for the proposed solution. First let us have  two more observations: 
\begin{itemize}
\item  Given the full knowledge of $F_1, F_2, F_3$, the inputs to each server in different queries are independent.  This means that there is no possibility of reverse computing and inferring about the order of computation. 

\item 
In all six options for $\sigma$, the order of computation in server one is $\underbrace{F_1 \to F_3\to F_1 \to\cdots \to F_1}_{2M+1}$, and the order of computation in server two $\underbrace{F_2\to F_3\to F_2 \to \cdots \to F_2}_{2M+1}$. Since the sequence of the computed functions by each server is predetermined and fixed, and is \emph{not} a function of $\sigma$, it cannot reveal any information about the  desired order of computations. 
\end{itemize}

We summarize the above six solutions for the six different options for $\sigma$ as follows: 
\begin{eqnarray*}
\begin{array}{|c|cc|} 
\multicolumn{3}{c}{} \\ \hline
\mbox{\tiny{BLOCK}}&\mbox{\tiny SERVER $1$}&\mbox{\tiny SERVER $2$}\\\hline
1&
 \begin{array}{l}
 F_1\\
 F_3
\end{array}
&
\begin{array}{l}
F_2 \\
F_3 
\end{array}\\ \hline
2&
 \begin{array}{l}
 F_1\\
 F_3
\end{array}
&
\begin{array}{l}
F_2 \\
F_3 
\end{array}\\ \hline
$\vdots$&
$~~~~~\vdots$
& 
$~~~~~\vdots$\\ \hline
M+1&
 \begin{array}{l}
 F_1
\end{array}
&
\begin{array}{l}
F_2
\end{array}\\ \hline
\end{array} 
\end{eqnarray*}


After showing that the proposed scheme has the privacy and correctness properties, we  compute its rate.  In total, there are $4M+2$  queries, while the user wants to compute three functions on $M$ input vectors. Hence, the rate of computation is $\frac{3M}{4M+2}$. As $M \to \infty$, this rate achieves the claimed lower bound on the capacity of the problem, which is equal to $\frac{3}{4}$.
\end{example}

\begin{example}[$K=4$ and  $N=3$]\label{exmpl_3} \normalfont
Example \ref{exmpl_2}  shows that one may use a  predetermined sequence of functions to be computed at each  server, and also use some random vectors, in order to achieve the  lower bound of Theorem~\ref{thrm1} on the capacity in the asymptotic regimes of large $M$. However, we require a systematic  approach to generalize that idea for arbitrary $K,N$. We explain the proposed approach in this  example. Note that this method is slightly different from the previous example.

Consider the problem of private sequential function computation with $N=3$ servers and $K=4$ functions\footnote{
In this example, we apply minor changes on the notations described in Section \ref{prob_state}. These changes are explained through the example.
}. Assume that the user has $M=2M'$ input vectors, denoted by $W_{m,j}$, $m \in [1:M], j \in \{1,2\}$. The reason that we use two indices $(m,j)$ for the inputs is to simplify the explanation.  In addition, we assume that the user wants to compute $F_{\sigma_4}(F_{\sigma_3}(F_{\sigma_2}(F_{\sigma_1}W_{1:M,1:2})))$, i.e., $\sigma=({\sigma_4}~{\sigma_3}~{\sigma_2}~ {\sigma_1})$.

We first  assign a fixed order of computations to each server. In particular, we assign the following order of computation to server  $n$ , $n \in \{1,2,3\}$, 
 \begin{align}
\underbrace{F_n\to F_n \to F_4  \to F_n \to F_n \to F_4 \to \cdots \to F_n\to F_n \to F_4}_{3(M+3)}\nonumber,
\end{align}
no matter what the desired order of composition $\sigma=({\sigma_4}~{\sigma_3}~{\sigma_2}~ {\sigma_1})$ is. 
This means that to compute $M=2M'$ requests, the user asks $3(M'+3)$ queries from each server, and $9M'+27$ queries in total. We call this step as \textit{function assignment}.  We illustrate the proposed \textit{function assignment}, universally for all possible $\sigma$,   as follows:

\begin{eqnarray*}
\begin{array}{|c|ccc|} 
\multicolumn{4}{c}{} \\ \hline
\mbox{\tiny{BLOCK}}&\mbox{\tiny SERVER $1$}&\mbox{\tiny SERVER $2$}&\mbox{\tiny SERVER $3$}\\\hline
1&
\begin{array}{l}
 F_1\\
 F_1\\
 F_4
\end{array}
&
\begin{array}{l}
F_2 \\
F_2\\
F_4
\end{array}&
\begin{array}{l}
F_3 \\
F_3\\
F_4
\end{array}\\ \hline
2&
\begin{array}{l}
 F_1\\
 F_1\\
 F_4
\end{array}
&
\begin{array}{l}
F_2 \\
F_2\\
F_4
\end{array}&
\begin{array}{l}
F_3 \\
F_3\\
F_4
\end{array}\\ \hline
$\vdots$&
$\vdots$
&
 $\vdots$ &
 $\vdots$  \\ \hline
M'+3&
\begin{array}{l}
 F_1\\
 F_1\\
 F_4
\end{array}
&
\begin{array}{l}
F_2 \\
F_2\\
F_4
\end{array}&
\begin{array}{l}
F_3 \\
F_3\\
F_4
\end{array}\\ \hline
\end{array}
\end{eqnarray*}

For the convenience in the description of the achievable scheme, we divide the computations of $3(M'+3)$ functions at each server to $M'+3$ \textit{blocks of computations}, each comprised of three computations per server. 
It can be observed that the function assignment is the same for all the blocks. Since the servers are not synchronized, we need to also assign an order for asking the queries based on the proposed function assignment. 
In the proposed scheme, we assume that the blocks of computations are requested by the user sequentially. This means that the user first asks all of the queries in a block from the servers, then begins the next block.  
Also assume that the queries in each block are asked by the user in an arbitrary order. 
We claim that by exploiting this function assignment and order determination for asking queries, one can design an achievable scheme for any desired permutation. 
Note that in this case, there are $4!=24$ distinct permutations. As examples, we illustrate the achievable scheme on two specific permutations.

\subsubsection{$\sigma = (1~ 3~ 4~2)$}

In the following figure, we  propose the procedure of computations in the first four  blocks. In this figure, we use a new notation for some variables $Z_*$. There is an important remark about $Z_*$.

\begin{remark} \normalfont
In each block of the scheme, we observe variables $Z_*$ in seven different places. These seven variables in each block named as $Z_*$ are essentially \emph{different}, drawn \emph{randomly,  uniformly, and independently}  from each other and from all other random variables in the solution. However, since each of them are not appear in the scheme again,  and to avoid confusing the reader by introducing a long list of random variables, they are not denoted differently. We use these random variables to keep the list of queries to each server for different choices of $\sigma$ identically distributed. 
\end{remark} 

In the first block, the requests corresponded to  $W_{1,\{1,2\}}$ are considered and the  function $F_2$ is applied on them (at server $2$).  
In the second block, the user has access to  $F_2W_{1,\{1,2\}}$, and asks the servers one and server two to apply the  function $F_4$ on them. To run $F_4$, the user utilizes a randomly drawn vector to ensure the privacy of the computation. Also, the user in this block again asks from the second server to perform a similar task to the previous block on new vectors $W_{2,\{1,2\}}$. The third block is also similar. The  function $F_3$ is executed on  the requests corresponded to 
$W_{1,\{1,2\}}$, the function $F_4$ is computed for the requests corresponded to $W_{2,\{1,2\}}$, and the  function $F_2$ is computed for the requests corresponded to $W_{3,\{1,2\}}$. 
In the fourth block, the user again asks three servers to compute specific functions similar to the third block. In addition, for the requests corresponded to $W_{1,\{1,2\}}$, the  function $F_1$ is computed and the final result of computation for them is available at the end of this block.  
The rest of the scheme is similar, where $M=2M'$ requests are computed in the $M'+3$ blocks of computations.  

Observe that the scheme is correct and private. The privacy is due to the fact that all the inputs given to each server are uniformly and independently drawn vectors, which are independent from the desired permutation of the user. Also the order of computations at each server is the same for all desired permutations. This can be confirmed by comparing the schemes for $\sigma=(1~3~4~2)$ and $\sigma=(4~3~2~1)$ (see Subsection~\ref{subsec4321}). To derive the scheme for other choices of $\sigma$, refer to the general achievable scheme.

\begin{eqnarray*}
\begin{array}{|c|ccc|} 
\multicolumn{4}{c}{\sigma=(1~3~4~2)} \\ \hline
\mbox{\tiny{BLOCK}}&\mbox{\tiny SERVER $1$}&\mbox{\tiny SERVER $2$}&\mbox{\tiny SERVER $3$}\\\hline
$1$&
\begin{array}{l}
 F_1 {\color{red} Z_*}\\
 F_1 {\color{red} Z_*}\\
 F_4   {\color{red} Z_*}
\end{array}
&
\begin{array}{l}
F_2 {\color{blue} W_{1,1}} \\
F_2 {\color{blue}W_{1,2}} \\
F_4   {\color{red} Z_*}
\end{array}&
\begin{array}{l}
F_3 {\color{red} Z_*}\\
F_3  {\color{red} Z_*}\\
F_4  {\color{red} Z_*}
\end{array}\\ \hline
$2$&
\begin{array}{l}
 F_1 {\color{red} Z_*}\\
 F_1 {\color{red} Z_*}\\
F_4 ( {\color{blue} F_2  { W_{1,1}}}\oplus {\color{red} Z_2})
\end{array}
&
\begin{array}{l}
F_2 {\color{PineGreen} W_{2,1}} \\
F_2 {\color{PineGreen}W_{2,2}} \\
F_4 ( {\color{blue} F_2  { W_{1,2}}}\oplus {\color{red} Z_2})
\end{array}&
\begin{array}{l}
F_3 {\color{red} Z_*}\\
F_3  {\color{red} Z_*}\\
F_4  {\color{red} Z_2} 
\end{array}\\ \hline
$3$&
\begin{array}{l}
 F_1 {\color{red} Z_*}\\
 F_1 {\color{red} Z_*}\\
F_4 ( {\color{PineGreen} F_2  { W_{2,1}}}\oplus {\color{red} Z_3})
\end{array}
&
\begin{array}{l}
F_2 {\color{Plum} W_{3,1}} \\
F_2 {\color{Plum}W_{3,2}} \\
F_4 ( {\color{PineGreen} F_2  { W_{2,2}}}\oplus {\color{red} Z_3})
\end{array}&
\begin{array}{l}
F_3 { \color{blue}(F_4(F_2W_{1,1} ))) }\\
F_3 { \color{blue}(F_4(F_2W_{1,2} ))) }\\
F_4  {\color{red} Z_3} 
\end{array}\\ \hline
$4$&
\begin{array}{l}
 F_1 {\color{blue} F_3(F_4(F_2 W_{1,1}))}\\
 F_1 {\color{blue} F_3(F_4(F_2 W_{1,2}))}\\
F_4 ( {\color{Plum} F_2  { W_{3,1}}}\oplus {\color{red} Z_4})
\end{array}
&
\begin{array}{l}
F_2 {\color{BrickRed} W_{4,1}} \\
F_2 {\color{BrickRed}W_{4,2}} \\
F_4 ( {\color{Plum} F_2  { W_{3,2}}}\oplus {\color{red} Z_4})
\end{array}&
\begin{array}{l}
F_3 { \color{PineGreen}(F_4(F_2W_{2,1} ))) }\\
F_3 { \color{PineGreen}(F_4(F_2W_{2,2} ))) }\\
F_4  {\color{red} Z_4} 
\end{array}\\ \hline
 \end{array}\\
\end{eqnarray*}

Now we compute the rate of the proposed scheme. Note that there are $2M'$ input vectors and we have $4$ functions ($8M'$ in total), while the user utilizes  the servers for $9M'+27$ times. 
Therefore, the rate of the proposed achievable scheme is $\frac{8M'}{9M'+27}$, which goes to $\frac{8}{9}$ as $M \to \infty$\footnote{
We note that this rate is not necessarily   optimum for the one-shot case and one may  compute privately $2M'$ requests with less than $9M'+27$ computations. However, in the asymptotic regime, the rate achieves the claimed lower bound on the capacity and the gap is vanishing. 
}. 
%

\subsubsection{$\sigma=(4~3~2~1)$}
\label{subsec4321}

To observe that the proposed approach works for any permutation, we consider another permutation here and construct a similar scheme for this case. The proposed scheme has been detailed in the following table.

\begin{eqnarray*}
\begin{array}{|c|ccc|} 
\multicolumn{4}{c}{\sigma=(4~3~2~1)} \\ \hline
\mbox{\tiny{BLOCK}}&\mbox{\tiny SERVER $1$}&\mbox{\tiny SERVER $2$}&\mbox{\tiny SERVER $3$}\\\hline
$1$&
\begin{array}{l}
 F_1{\color{blue} W_{1,1}}\\
 F_1{\color{blue}W_{2,1}}\\
 F_4   {\color{red} Z_*}
\end{array}
&
\begin{array}{l}
F_2  {\color{red} Z_*} \\
F_2 {\color{red} Z_*} \\
F_4   {\color{red} Z_*}
\end{array}&
\begin{array}{l}
F_3 {\color{red} Z_*}\\
F_3  {\color{red} Z_*}\\
F_4  {\color{red} Z_*}
\end{array}\\ \hline
$2$&
\begin{array}{l}
 F_1{\color{PineGreen} W_{2,1}}\\
 F_1{\color{PineGreen}W_{2,2}}\\
 F_4   {\color{red} Z_*}
\end{array}
&
\begin{array}{l}
F_2  {\color{blue} (F_1W_{1,1})} \\
F_2  {\color{blue} (F_2W_{1,2})} \\
F_4   {\color{red} Z_*}
\end{array}&
\begin{array}{l}
F_3 {\color{red} Z_*}\\
F_3  {\color{red} Z_*}\\
F_4  {\color{red} Z_*}
\end{array}\\ \hline
$3$&
\begin{array}{l}
 F_1{\color{Plum} W_{3,1}}\\
 F_1{\color{Plum}W_{3,2}}\\
 F_4   {\color{red} Z_*}
\end{array}
&
\begin{array}{l}
F_2  {\color{PineGreen} (F_1W_{2,1})} \\
F_2  {\color{PineGreen} (F_2W_{2,2})} \\
F_4   {\color{red} Z_*}
\end{array}&
\begin{array}{l}
F_3 {\color{blue} ( F_2  { (F_1W_{1,1})})}\\
F_3 {\color{blue} ( F_2  { (F_1W_{1,2})})}\\
F_4  {\color{red} Z_*}
\end{array}\\ \hline
$4$&
\begin{array}{l}
 F_1{\color{BrickRed} W_{4,1}}\\
 F_1{\color{BrickRed}W_{4,2}}\\
 F_4 ( {\color{blue}F_3  ( F_2  { (F_1W_{1,1})})}\oplus {\color{red} Z_4})
\end{array}
&
\begin{array}{l}
F_2  {\color{Plum} (F_1W_{3,1})} \\
F_2  {\color{Plum} (F_2W_{3,2})} \\
F_4 ( {\color{blue}F_3  ( F_2  { (F_1W_{1,2})})}\oplus {\color{red} Z_4})
\end{array}&
\begin{array}{l}
F_3 {\color{PineGreen} ( F_2  { (F_1W_{2,1})})}\\
F_3 {\color{PineGreen} ( F_2  { (F_1W_{2,2})})}\\
F_4  {\color{red} Z_4}
\end{array}\\ \hline
 \end{array}
\end{eqnarray*}

We now review the proposed scheme in Example \ref{exmpl_3}, and define some new notations.  The new notations will help us to generalize it in the next section.  We note that for any set of vectors $W_{m,\{1,2\}}$, we first need to compute $F_{\sigma_1}$, then $F_{\sigma_2}$, up to $F_{\sigma_K}$. Let us denote formally the process of computation of the $k^{th}$ function $F_{\sigma_k}$  for the requests corresponded to $W_{m,\{1,2\}}$, by $\calR_m^k$. 
As discussed, in  Example \ref{exmpl_3},  the scheme  utilizes $M'+3$  blocks of computations,  and  performs all the tasks $\calR_m^k$, for $k\in \{1,2,3,4\}$ and $m \in [1:M']$.

In the formal description in above, we assigned the task $\calR_1^1$ to the first block. This means that at the end of the first block,  $F_{\sigma_1}W_{1,\{1,2\}}$ is available at the user side. 
For the second block, we assigned two tasks $\calR_2^1$ and $\calR_1^2$. This means that at the end of this block, the user has access to  $F_{\sigma_1}W_{2,\{1,2\}}$ and $F_{\sigma_2}(F_{\sigma_1}(W_{1,\{1,2\}}))$. 
Generally, before  the $m^{th}$ block, the tasks $\calR_{m'}^{k'}$, for $k'+m' \le m $, have already been  performed, and the tasks  $\calR_{m'}^{k'}$, for $k'+m' = m+1 $,  will be performed at the $m^{th}$ block. By this explanation, one can see that if there is a computation scheme that can perform all of the above tasks, then it is $2M'-$achievable. 
In  Example \ref{exmpl_3}, it is shown that such  computation scheme   exists. 
We will  generalize this approach for arbitrary $K,N$  in the next section to construct an achievable scheme\footnote{
We used the notion of task here only to explain our achievable scheme more apprehensible. We note that there is not any necessity for the user to achieve the outputs of the  tasks. The only necessity is to find the target result with high accuracy; for more explanations see Section \ref{prob_state}.  
}.

\end{example}

 \section{The Achievable Scheme}\label{achieve}

 \subsection{Preliminaries}
 
 In this section, we propose an achievable scheme for arbitrary $K,N$. We note that as shown in Example \ref{exmpl_1}, one can simply design an achievable scheme when $K \le N$. Hence, throughout this section, we focus on the cases with  $K > N$. 
Also, with a minor change in  assumptions described in Section \ref{prob_state}, we assume that there are $M = M'(N-1)$ requests in the system, i.e., the number of requests is divisible by $N-1$ . In other words, to propose an achievable scheme, we assume that the number of requests  are divisible by $N-1$. Later, we will argue that this assumption does not affect the the rate in the asymptotic regimes. 
We further divide the given $M=M'(N-1)$ input vectors  to $M'$ batches of $N-1$ input vectors. For each $m \in [1:M']$, we denote the input vectors in the $m^{th}$ batch by  $W_{m,j}$, $j\in [1:N-1]$.

To ensure the privacy and correctness constraints, we rely on the following two ideas: 
\begin{itemize}
\item The index of functions to be computed by each server is assigned in a deterministic manner, no matter what permutation is to be computed. We refer to this task as \textit{function assignment} in the paper. We use the notion of \textit{block of computations}, similar to Section \ref{motivation}, to propose an appropriate function assignment. 

\item
Whenever required, the user exploits randomly generated vectors to ensure  privacy.
\end{itemize}

In the following subsections, we first propose a deterministic function assignment, and then we design an appropriate task assignment. After these steps, we propose a \textit{vector assignment}, which is defined as the process of  associating  appropriate vectors to be transmitted from the user to the servers in queries. 

 \subsection{Function Assignment}

First we define the notion of \textit{block of computations} as follows.
 
 \begin{definition}
 A block of the computations for a private sequential function computation problem with $N$ servers and $K$ functions is defined as

\begin{eqnarray*}
\begin{array}{|c|ccccc|} \hline
 \mbox{\tiny{PHASE}}& \mbox{\tiny{NUMBER}} & \mbox{\tiny ~~SERVER $1$}&\mbox{\tiny ~~SERVER $2$}&\mbox{\tiny $\cdots$}&\mbox{\tiny ~~SERVER $N$}\\\hline
 1 & 
\begin{array}{c}
1\\
2\\
 $\vdots$\\
N-1
\end{array}&\cellcolor{yellow!35}
\begin{array}{c}
 F_1\\
 F_1\\
 $\vdots$\\
 F_{1} 
 \end{array} &\cellcolor{yellow!35}
\begin{array}{c}
F_2 \\
F_2\\
$\vdots$\\
F_2
\end{array} &\cellcolor{yellow!35}
\begin{array}{c}
\cdots \\
\cdots \\
\ddots \\
\cdots
\end{array}&\cellcolor{yellow!35}
\begin{array}{c}
F_N\\
F_{N}\\
$\vdots$\\
F_N
\end{array} \\ \hline 
 $2$  & 
\begin{array}{c}
 N\\
N+1\\
$\vdots$\\
K-1
\end{array}& \cellcolor{orange!35}
\begin{array}{c}
 F_{N+1}\\
  F_{N+2}\\
  $\vdots$\\
   F_{K}
\end{array}& \cellcolor{orange!35}
\begin{array}{c}
 F_{N+1}\\
  F_{N+2}\\
  $\vdots$\\
   F_{K}
\end{array} & \cellcolor{orange!35}
\begin{array}{c}
\cdots \\
\cdots \\
\ddots \\
\cdots 
\end{array}&  \cellcolor{orange!35}
\begin{array}{c}
 F_{N+1}\\
  F_{N+2}\\
  $\vdots$\\
   F_{K}
\end{array}\\ \hline
 \end{array}
\end{eqnarray*}
%
As shown above, each block contains two phases. In the first phase, each server computes a  specific function, for $N-1$ times, and in the second phase, all servers compute $K-N$similar functions. 
 \end{definition}
 
Now, we are ready to propose the function assignment of the achievable scheme. In the achievable scheme, we utilize a deterministic function assignment, which comprised of $M'+K-1$ replicated blocks of computations a described above. To determine the order of queries asked by the user, we use the following rules:
 \begin{itemize}
\item The queries are asked block by block, i.e., the user first asks all of the queries of the first block, then  asks all of the queries of the second block, and so on. 
\item At each block, the user asks the queries by an arbitrary order. As we will show, all the  vectors send by the user to the servers at each block are available at the user side before the block begins.

\end{itemize}

 \subsection{Vector  Assignment}\label{inputassign}

 Let us first introduce some useful  notations. 
 
 \begin{definition}
 Define
 \begin{align}
\In_i(\calR_m^k):= F_{\sigma_{k-1}}(F_{\sigma_{k-2}}(\cdots (F_{\sigma_{1}}W_{m,i})\cdots))
\end{align} and
\begin{align}
\out_i(\calR_m^k):= F_{\sigma_{k}}(F_{\sigma_{k-1}}(\cdots (F_{\sigma_{1}}W_{m,i})\cdots))
\end{align}  
for any  $m\in [1:M'],k \in [1:K]$, and $i \in [1:N-1]$.  In other words, we denote the procedure of computation of the $k^{th}$ function in the sequential computation  (i.e.,  $F_{\sigma_k}$) for the requests corresponded to the $m^{th}$ batch by $\calR_m^k$, and the inputs and outputs for this task are $\In_i(\calR_m^k)$, $\out_i(\calR_m^k)$, respectively.

 \end{definition}

 \begin{corollary}
 \begin{align}
\out_i(\calR_m^k) &=F_{\sigma_k} \In_i(\calR_m^k)\\
& = \In_i(\calR_m^{k+1}). 
\end{align}
 \end{corollary}
Clearly,  the task $\calR_m^k$ is said to be performed whenever the value  of $\out_{1:N-1}(\calR_m^k)$ is available at the user side.

The last step to explain the achievable scheme is to define the vectors transmitted from the user to the servers in each step.  Let $\pi$ denote the unique permutation $\pi \in \calS_K$ such that $\sigma \pi = \pi \sigma $ is the identity permutation. In the proposed achievable scheme,  for each $m \in [1:M'+K-1]$, at the $m^{th}$ block of computations, the user asks the $n^{th}$ server to perform the computation on the following vectors:

\begin{center}
\begin{tabular}{|c|c|cc|}
  \hline
  \mbox{\tiny BLOCK} & \mbox{\tiny PHASE} &  \mbox{\tiny SERVER $1\le n \le N-1$}  &  \mbox{\tiny SERVER $N$} \\ \hline 
  $\vdots$ &$\vdots$ &$\vdots$ &$\vdots$ \\ \hline
 \multirow{2}[15]{*}{$m$}& 1 & 
 {\cellcolor{yellow!35} 
$\begin{array}{c} 
{\color{black} \In_1(\calR_{m-\pi_n+1}^{\pi_n})}\\
 {\color{black} \In_2(\calR_{m-\pi_n+1}^{\pi_n})}\\
 \vdots \\
 {\color{black} \In_{N-1}(\calR_{m-\pi_n+1}^{\pi_n})}
\end{array}$
} & {\cellcolor{yellow!35} 
$\begin{array}{c} 
{\color{black} \In_1(\calR_{m-\pi_N+1}^{\pi_N})}\\
 {\color{black} \In_2(\calR_{m-\pi_N+1}^{\pi_N})}\\
 \vdots \\
 {\color{black} \In_{N-1}(\calR_{m-\pi_N+1}^{\pi_N})}
\end{array}$
} \\ \cline{2-2} \cline{3-2}
 & 2 &  \cellcolor{orange!35}
  $\begin{array}{c} 
 {\color{black} \In_n(\calR_{m-\pi_{N+1}+1}^{\pi_{N+1}})} \oplus {\color{red}  Z_{t,1}}  \\ 
 {\color{black} \In_n(\calR_{m-\pi_{N+2}+1}^{\pi_{N+2}})} \oplus {\color{red}  Z_{t,2}}  \\
  \vdots \\
   {\color{black} \In_n(\calR_{m-\pi_{K}+1}^{\pi_{K}})} \oplus {\color{red}  Z_{m,K-N}} 
 \end{array}$ &  \cellcolor{Dandelion!25}
   $\begin{array}{c}
  {\color{red}  Z_{m,1}}  \\ 
  {\color{red}  Z_{m,2}}  \\
  \vdots \\
     {\color{red}  Z_{m,K-N}} 
 \end{array}$ 
  \\ \hline
    $\vdots$ &$\vdots$ &$\vdots$ & $\vdots$  \\ \hline
\end{tabular}
\end{center}
In above, at the first phase, the server $n$, receives the vectors $\In_i(\calR^{\pi_n}_{m-\pi_n+1})$, $i \in [1:N-1]$. In the second phase, the two cases $n<N$ and $n=N$ are different. If $n<N$, then the server receives $\In_n(\calR_{m-\pi_{N+i}+1}^{\pi_{N+i}}) \oplus  Z_{m,i}$, $i\in [1:K-N]$,  and if $n=N$, it receives $ Z_{m,i}$, $i\in [1:K-N]$. 
Note that the random vectors $Z_{m,i}$, $i \in [1:K-N]$, are distributed uniformly over $\mathbb{F}^L$, and they are independent from each other and from all the other random variables in the problem.
 
We note that the vectors $\In_i(\calR^{\pi_k}_{m-\pi_k+1})$ for each $i\in [1:N-1]$, $k\in [1:K]$, and $t\in [1:M'+K-1]$ are well defined, except for the cases that $m-\pi_k+1 \le 0 $.  For such cases, the user utilizes a uniformly drawn random vector, which is independent from all of the other variables in the problem, rather than  $\In_i(\calR^{\pi_k}_{m-\pi_k+1})$, in the vector assignment. See the vectors $Z_*$ in  Example \ref{exmpl_3} for more details. 

An important question regarding the proposed vector assignment is that whether it is feasible or not. More precisely, it is required to be proved that the inputs assigned to each block are available at the end of the previous block.  In the next section, we will prove the propose vector assignment is indeed feasible.

 In Figure \ref{fig:input}, we illustrate the proposed function assignment and vector assignment.

\begin{sidewaysfigure} 
\begin{eqnarray*}
\begin{array}{|c|ccccc|} 
\multicolumn{6}{c}{\sigma=(\sigma_K~\sigma_{K-1}~\cdots~\sigma_1),~ \sigma^{-1} = (\pi_K~\pi_{K-1}~\cdots~\pi_1) } \\ \hline
\mbox{\tiny{BLOCK}}&\mbox{\tiny SERVER $1$}&\mbox{\tiny \hspace{0.5cm} SERVER $2$}&\mbox{\tiny $\cdots$}&\mbox{\tiny SERVER $N-1$}&\mbox{\tiny SERVER $N$}\\\hline
$\vdots$&
$\vdots$
&
 $\vdots$ &
 $\vdots$ & 
 $\vdots$&
$\vdots$ \\ \hline
$m$&
\begin{array}{l}
 F_1{\color{blue} \In_1(\calR_{m-\pi_1+1}^{\pi_1})}\\
F_1{\color{blue} \In_2(\calR_{m-\pi_1+1}^{\pi_1})}\\
 \vdots \\ 
 F_1{\color{blue} \In_{N-1}(\calR_{m-\pi_1+1}^{\pi_1})}\\ 
 F_{N+1}  ( {\color{blue} \In_1(\calR_{m-\pi_{N+1}+1}^{\pi_{N+1}})} \oplus {\color{red}  Z_{m,1}}) \\ 
  F_{N+2}  ( {\color{blue} \In_1(\calR_{m-\pi_{N+2}+1}^{\pi_{N+2}})} \oplus {\color{red}  Z_{m,2}}) \\ 
  \vdots \\
   F_{K}  ( {\color{blue} \In_1(\calR_{m-\pi_{K}+1}^{\pi_{K}})} \oplus {\color{red}  Z_{m,K-N}})
\end{array}
&
\begin{array}{l}
 F_2{\color{PineGreen} \In_1(\calR_{m-\pi_2+1}^{\pi_2})}\\ 
F_2{\color{PineGreen} \In_2(\calR_{m-\pi_2+1}^{\pi_2})}\\
 \vdots \\
 F_2{\color{PineGreen} \In_{N-1}(\calR_{m-\pi_2+1}^{\pi_2})}\\ 
 F_{N+1}  ( {\color{PineGreen} \In_2(\calR_{m-\pi_{N+1}+1}^{\pi_{N+1}})} \oplus {\color{red}  Z_{m,1}}) \\ 
  F_{N+2}  ( {\color{PineGreen} \In_2(\calR_{m-\pi_{N+2}+1}^{\pi_{N+2}})} \oplus {\color{red}  Z_{m,2}}) \\ 
  \vdots \\ 
   F_{K}  ( {\color{PineGreen} \In_2(\calR_{m-\pi_{K}+1}^{\pi_{K}})} \oplus {\color{red}  Z_{m,K-N}})
\end{array}
&
\begin{array}{l}
\vdots \\
\end{array}&
\begin{array}{l}
 F_{N-1}{\color{purple} \In_1(\calR_{m-\pi_{N-1}+1}^{\pi_{N-1}})}\\ 
F_{N-1}{\color{purple} \In_2(\calR_{m-\pi_{N-1}+1}^{\pi_{N-1}})}\\ 
 \vdots \\
 F_{N-1}{\color{purple} \In_{N-1}(\calR_{m-\pi_{N-1}+1}^{\pi_{N-1}})}\\
 F_{N+1}  ( {\color{purple} \In_{N-1}(\calR_{m-\pi_{N+1}+1}^{\pi_{N+1}})} \oplus {\color{red}  Z_{m,1}}) \\
  F_{N+2}  ( {\color{purple} \In_{N-1}(\calR_{m-\pi_{N+2}+1}^{\pi_{N+2}})} \oplus {\color{red}  Z_{m,2}}) \\ 
  \vdots \\ 
   F_{K}  ( {\color{purple} \In_{N-1}(\calR_{m-\pi_{K}+1}^{\pi_{K}})} \oplus {\color{red}  Z_{m,K-N}})
\end{array}&
\begin{array}{l}
 F_{N}{\color{Blue} \In_1(\calR_{m-\pi_{K}+1}^{\pi_{K}})}\\
F_{N}{\color{Blue} \In_2(\calR_{m-\pi_{K}+1}^{\pi_{K}})}\\
 \vdots \\
 F_{N}{\color{Blue} \In_{N-1}(\calR_{m-\pi_{K}+1}^{\pi_{K}})}\\ 
 F_{N+1}  {\color{red}  Z_{m,1}} \\
  F_{N+2}   {\color{red}  Z_{m,2}} \\ 
  \vdots \\ 
   F_{K}  {\color{red}  Z_{m,K-N}}  
\end{array}\\ \hline
$\vdots$&
$\vdots$
&
 $\vdots$ &
 &
 $\vdots$
 &
 $\vdots$  \\ \hline
 \end{array}
\end{eqnarray*}
\caption{The proposed function assignment and vector assignment.}
\label{fig:input}
\end{sidewaysfigure}

 \section{Achievability Proof of Theorem \ref{thrm1}}\label{proof_achieve}

 In this section, we prove the achievability part of Theorem \ref{thrm1}. In particular, we focus on the proposed scheme described in the previous section.  We first prove that this scheme is feasible, i.e., the proposed vector assignment is feasible, meaning that the user has access to the contents required to be transmitted to the servers at the time of transmission (see Subsection \ref{inputassign}).  Then, we prove the correctness and privacy of the proposed scheme. 
 
Finally, we compute the rate of the achievable scheme,  and then we  prove the desired result. 
 
 \subsection{Proof of Feasibility and Correctness}
 
 In this part, we first prove that the proposed vector assignment is feasible to be performed by the user, and then, we establish the correctness of the scheme. 
 
 \begin{lemma}\label{lmma1}
 The proposed vector assignment is feasible, i.e,  the vectors  which are to be sent by the user at each block of computations to the servers can be obtained by computing a linear combination of the available outputs of the previous blocks. 
 \end{lemma}
\begin{proof}

 To prove the feasibility of the vector assignment, we use induction. Note that, the vector assignment of block one is feasible trivially. Here we show that if the vector assignment of block $m$ is feasible, the vector assignment of block $m+1$ is also feasible. 
 
Since vector assignment of block $m$ is feasible,  $\In_{1:N-1}(\calR_{m-\pi_k+1}^{\pi_k})$, $k \in [1:K]$, are available before the $m^{th}$ block of computations begins. We need to show that $\In_{1:N-1}(\calR_{(m+1)-\pi_k+1}^{\pi_k})$ for each  $k \in [1:K]$ can be obtained from the information provided to the user until the end of the $m^{th}$ block. 

From the induction hypothesis,  the user can ask the queries based on the structure described in vector assignment step. We note that after the first phase of the block, the values of $F_{n'} \In_{i}(\calR_{m-\pi_{n'}+1}^{\pi_{n'}})$, $i \in [1:N-1], {n'} \in [1:N]$ are available at the user side. In addition, at the second phase of the $m^{th}$ block, if we cancel the random vectors $Z_{m,i}$, $i \in [1:K-N]$, we conclude that all of the values of $F_{n'}\In_i(\calR^{\pi_{n'}}_{m-\pi_{n'}+1})$, $i \in [1:K-1]$, $n' \in [N+1:K]$ are available at the user side. Thus,  all the values of $F_{n'}\In_i(\calR^{\pi_{n'}}_{m-\pi_{n'}+1})$, $i \in [1:N-1]$, $n' \in [1:K]$, are available for the user after the completion of the $m^{th}$ block. 

Now observe that by changing the variables as $n' = \sigma_k$,
\begin{align}
F_{n'} \In_{i}(\calR_{m-\pi_{n'}+1}^{\pi_{n'}}) &=  F_{\sigma_k} \In_{i}(\calR_{m-\pi_{\sigma_k}+1}^{\pi_{\sigma_k}}) \\
&=  F_{\sigma_k} \In_{i}(\calR_{m-k+1}^{k}) \\
& = \out_{i}(\calR_{m-k+1}^{k}),
\end{align}
which shows that after the $m^{th}$ block, the tasks $\calR^{k'}_{m'}$ such that $m'+k' = m+1$ are performed.  Note that\footnote{
If ${(m+1)-\pi_k+1}\le 0$, then the claim is trivial, because the required vectors are randomly drawn. 
}  $\In_{i}(\calR_{(m+1)-\pi_k+1}^{\pi_k})= \out_{i}(\calR_{(m+1)-\pi_k+1}^{\pi_k-1})= \out_{i}(\calR_{m-(\pi_k-1)+1}^{\pi_k-1})$ for any $k \in [1:K]$ and $i \in [1:N-1]$.  Hence, the proof is complete. 
 \end{proof}
 
Note that the proof of Lemma \ref{lmma1}, we can conclude that at the end of the $(M'+K-1)^{th}$ block, all the tasks $\calR^K_m$, $m\in [1:M']$ are performed, and hence the values of $\out_{1:N-1}(\calR^K_m)$ are available at the user side. Therefore, the proposed scheme is also correct, meaning that after running the proposed scheme, the user achieves its desired result (see the correctness condition in Section \ref{prob_state} for more explanations).

   \subsection{Proof of Privacy}
   To prove the privacy constraint, we require to show that  in the proposed achievable scheme, $I(\tilde{Q}_n,F_{1:K};\Sigma) = 0$, for each $n \in [1:N]$. Due to the deterministic function assignment in the proposed scheme, the only requirement is to show that the inputs given to the servers do not leak any information about the desired  permutation. 
   
Let us denote the inputs given to the $n^{th}$ server at the $m^{th}$ block of computations by $X_{n,m}  := \In_{1:N-1}(\calR_{m-\pi_n+1}^{\pi_n})$, for the first phase, and $Y_{n,m,[1:K-N]} \oplus Z_{m,[1:K-N]}$, for the second phase.  
Due to the proposed vector assignment, we have
$$ Y_{n,m} := Y_{n,m,[1:K-N]} = \left\{
  \begin{array}{ll}
    \In_{n}(\calR_{m-\pi_{[N+1:K]}+1}^{\pi_{[N+1:K]}}) & : n \in [1:N-1]\\
    O_{[N+1:K]} & : n = N,
  \end{array}
\right.
$$
where $O_i \in \mathcal{F}^L$, $i \in [K+1:N]$, are all zero vectors. 
In addition, we define $Z_m:=Z_{m,[1:K-N]}$ and we briefly write $Y_{n,m} \oplus Z_{m}$ to denote $Y_{n,m,[1:K-N]} \oplus Z_{m,[1:K-N]}$.

We need to show that 
\begin{align}
\pbb(\Sigma=\sigma'|\tilde{Q}_n,F_{1:K}) = \pbb(\Sigma=\sigma''|\tilde{Q}_n,F_{1:K}),
\end{align}
for each $\sigma',\sigma'' \in \calS_K$. Note that 
\begin{align}
\pbb(\Sigma=\sigma|\tilde{Q}_n,F_{1:K}) &= \frac{\pbb(\tilde{Q}_n,F_{1:K}|\Sigma=\sigma)\pbb(\Sigma=\sigma)}{\pbb(\tilde{Q}_n,F_{1:K})} \\
& = \frac{\pbb(F_{1:K}|\Sigma=\sigma)\pbb(\tilde{Q}_n |\Sigma=\sigma,F_{1:K})\pbb(\Sigma=\sigma)}{\pbb(\tilde{Q}_n,F_{1:K})} \\
&=\mu(\tilde{Q}_n,F_{1:K})
 \pbb(\tilde{Q}_n|\Sigma=\sigma,F_{1:K}),
\end{align}
where $\mu(\tilde{Q}_n,F_{1:K})$  does not depend on $\sigma$.  %
    Also, 
\begin{align}
 \pbb(\tilde{Q}_n|\Sigma=\sigma,F_{1:K})& \overset{(a)}{=} \pbb \Big( (X_{n,m},Y_{n,m}\oplus Z_m)_{m \in [1:M'+K-1]} |\Sigma=\sigma,F_{1:K} \Big)\\
  & \overset{(b)}{=} \prod_{m \in [1:M'+K-1]} \pbb( X_{n,m},Y_{n,m}\oplus Z_m |\Sigma=\sigma,F_{1:K})\\
  & \overset{(c)}{=} \prod_{m \in [1:M'+K-1]} \pbb( X_{n,m}|\Sigma=\sigma,F_{1:K}) \prod_{m \in [1:M'+K-1]}  \pbb(Y_{n,m}\oplus Z_m |\Sigma=\sigma,F_{1:K})\\
  &  =  \frac{1}{|\mathbb{F}|^{L(N-1)(M'+K-1)}} \times  \frac{1}{|\mathbb{F}|^{L(K-N)(M'+K-1)}} \\
  &  = \frac{1}{|\mathbb{F}|^{L(K-1)(M'+K-1)}},  
 \end{align}
%
%
%
which completes the proof. Note that $(a)$ holds because the utilized function assignment is deterministic, $(b)$ holds because conditioned on  $\Sigma=\sigma$ and $F_{1:K}$, we have the following Markov chain:
\begin{align} 
(X_{n,m_1}, Y_{n,m_1}\oplus Z_{m_1}) \to   W_{m_1-\pi_n+1,1:N-1} \to  W_{m_2-\pi_n+1,1:N-1} \to  (X_{n,m_2}, Y_{n,m_2}\oplus Z_{m_2}),
\end{align}
for each\footnote{
Note that if $M_i-\pi_n+1\le 0$, then the claimed Independence hold trivially. 
} $m_1\neq m_2$. Also, $(c)$ holds because 
\begin{align}
I(Y_{n,m} \oplus Z_m;X_{n,m}|\Sigma = \sigma, F_{1:K})& =H(Y_{n,m} \oplus Z_m|\Sigma=\sigma, F_{1:K}) - H(Y_{n,m} \oplus Z_m|X_{n,m}, \Sigma=\sigma,F_{1:K})\\
&\le (K-N)L  - H(Y_{n,m}  \oplus Z_t|X_{n,m}, \Sigma=\sigma, F_{1:K})\\
&\le (K-N) L   - H(Y_{n,m} \oplus Z_t|X_{n,m}, \Sigma=\sigma,F_{1:K},Y_{n,m})\\
& = (K-N) L - H(Z_m)\\
& = (K-N) L   - (K-N) L  \\
&=0.
\end{align}

      \subsection{Proof of Achievability for the cases where $N-1 \not| M$}
     
     In the previous section, we provided an achievable scheme for the cases where the number of requests is divided by $N-1$. 
     However, to prove the achievability, we need to propose a sequence of achievable scheme for any number of requests (see Section \ref{prob_state}).

Consider a general case of the private sequential function computation problem in which there are arbitrary number of requests $M=M'(N-1) + r$. Here $0 \le r \le N-2$ is an arbitrary integer. 
We propose an achievable scheme for the problem with these parameters as follows. For the first $M'(N-1)$ requests, assume that  the user utilizes the proposed scheme of the previous section. For any  $r$ remaining  requests, the user asks an arbitrary server to compute all  of the possible permutations for that request, i.e., the user asks the server for $K\times K!$ times\footnote{
Actually, this is not efficient to ask this numerous number of requests and the user can ask fewer questions. However, this effect vanishes in the asymptotic regime. 
} for each request. Due to the previous discussions, it is obvious that this scheme is both private and correct. 
Hence, it gives a $(M'(N-1)+r)-$achievable rate. 

      Let us  compute the rate for the proposed scheme for arbitrary number of requests $M'(N-1)+r$. For the first $M'(N-1)$ requests, the scheme includes $M'+K-1$ blocks, each require $N(K-1)$ function computations. 
      For the remaining $r$ requests, the user asks the servers for $r\times K\times K!$ times. All in all, there are $(M'+K-1)\times N(K-1)+r\times K\times K!$ number of queries.       
      Also, there are $M'(N-1)+r$ requests for the sequential computation of $K$ function at the user side. Hence, the rate of the proposed scheme is
      \begin{align}
R = \frac{K \times (M'(N-1)+r)}{(M'+K-1) \times N(K-1)+r\times K\times K!}.
\end{align}
   One can see that as the number of requests tends to infinity (i.e., $M' \to \infty$), this rate achieves $\frac{K(N-1)}{N(K-1)}= \frac{1-\frac{1}{N}}{1-\frac{1}{K}}$, which matches the  capacity of the problem. Therefore, the achievability proof  is complete. 
   
  \section{Converse Proof of Theorem \ref{thrm1}}\label{proof_converse}
 
 We prove the converse of Theorem \ref{thrm1}
 in the following. 
 
 \subsection{Preliminaries}

  \begin{lemma}\label{lemma3}
 Consider a matrix
 $F\in \{F'\in \mathbb{F}^{L\times L}:\det(F')\neq 0 \}$, which is chosen randomly and uniformly, and  $M$ randomly and uniformly generated vectors, $W_{m}\in \mathbb{F}$, $m \in [1:M]$,  each independent from the other vectors and from $F$. 
 Let  $W := (W_1,W_2,\ldots,W_M)\in \mathbb{F}^{L\times M}$ be a randomly generated matrix. We claim that the following propositions hold:
 \begin{enumerate}[label=(\alph*)]
\item   $H(FW_{1:M}) = H(W_{1:M})=ML$
\item   $H(FW_{1:M}|W_{1:M}) =L \times  \mathbb{E} \Big [  \rank (W) \Big ]$
\item   $\pbb(\rank(W) <M)  \xrightarrow{ L  \to \infty } 0$
\item   $ H(FW_{1:M}|W_{1:M}) = M\times (L-o(L))$
\end{enumerate}
 \end{lemma}

\begin{proof}[Proof of (a)]
$ML =  H(W_{1:M}) = H(FW_{1:M}|F) \le H(FW_{1:M}) \le ML$.
\end{proof}

\begin{proof}[Proof of (b)]
Denote an arbitrary realization of random vectors $W_{1:M}$ by $w_{1:M}$. 
Let $\rank(w) = \dim (\Span(w_{\calS}))$, where $\calS = \{s_1,s_2,\ldots, s_{\rank(w)}\} \subseteq [1:M]$ includes $\rank(w)$ distinct elements.  
Add specific vectors to the columns of matrix $(w_{s_1},w_{s_2},\ldots,w_{s_{\rank(w)}})$, and construct a matrix 
\begin{align}
U =\Big  (w_{s_1},w_{s_2},\ldots,w_{s_{\rank(w)}},U_{1},U_2,\ldots,U_{L-\rank(w)}\Big) \in \mathbb{F}^{L\times L},
\end{align}
 such that $U$ is an invertible matrix.  Note that $F' := FU \sim F$.  Let $e_m \in \mathbb{F}^L$ denotes the vector with all zero elements, except the $m^{th}$ element, which is equal to one.   Note that $U^{-1}w_{s_{m}} = e_m$, for all $M \in [1:\rank(w)]$. 
 Now we write
 \begin{align}
H(FW_{1:M}|W_{1:M}=w_{1:M}) & 
= \sum_{m=1}^{\rank(w)} H(Fw_{s_m}|Fw_{s_{1:m-1}})+H(Fw_{[1:M]\backslash \calS }|Fw_{\calS})\\
&\overset{(a)}{=} \sum_{m=1}^{\rank(w)} H(Fw_{s_m}|Fw_{s_{1:m-1}})\\
 &=\sum_{m=1}^{\rank(w)} H((FU)U^{-1}w_{s_m}|(FU)U^{-1}w_{s_{1:m-1}})\\
  & =  \sum_{m=1}^{\rank(w)} H(F'U^{-1}w_{s_m}|F'U^{-1}w_{s_{1:m-1}})\\ 
   & =  \sum_{m=1}^{\rank(w)} H(F'e_m|F'e_{1:m-1})\\
      & = \sum_{m=1}^{\rank(w)} L\\
       &= L\times \rank(w),
\end{align}
where $(a)$ follows since $w_{m}\in \Span(w_{\calS})$, for each  $m \in [1:M] \backslash \calS$.
Now taking the expectation from the two sides of the above inequality yields the desired result. 
\end{proof}

\begin{proof}[Proof of (c)]
Observe that 
\begin{align}
\pbb(\rank(W) <M) &= \pbb \Big (\bigcup_{m=1}^{N} \{ W_m \in \Span (W_{1:m-1}) \} \Big ). \\
\end{align}
Note that 
$\dim (\Span (W_{1:m-1})) \le m-1$, 
which means that each vector in $ \Span (W_{1:m-1})$ can be written as the weighted summation of $m-1$ specific vectors. This means that $ \Span (W_{1:m-1})$ contains at most $|\mathbb{F}|^{m-1}$ distinct vectors. However, the vector $W_m$ is chosen uniformly from the set $\mathbb{F}^L$.
Therefore, the probability that $W_m$ lies  in $\Span (W_{1:m-1})$ can be upper bounded by $|\mathbb{F}|^{m-1-L}$. Now we write
 \begin{align}
\pbb \Big (\bigcup_{m=1}^{M} \{ W_m \in \Span (W_{1:m-1}) \} \Big ) &\le \sum_{m=1}^{M} \pbb \Big (  W_m \in \Span (W_{1:m-1}) \Big ) \\
&\le \sum_{m=1}^M |\mathbb{F}|^{m-1-L}\\
& = \frac{|\mathbb{F}|^M-1}{|\mathbb{F}|^L(|\mathbb{F}|-1)},\\
\end{align}
which goes to zero as $L\to \infty$. This completes the proof.
\end{proof}

\begin{proof}[Proof of (d)]
Using $(b),(c)$, we obtain 
\begin{align}
 H(FW_{1:M}|W_{1:M}) &= L \times \mathbb{E} \Big [  \rank (W) \Big ] \\
 &=L \times \Big (\sum_{0=1}^{M-1}m \times \pbb(\rank (W)=m)+M \times \pbb(\rank (W)=M)  \Big ) \\
&= L \times \Big ( \sum_{m=1}^{M-1}m \times o(1)+M \times (1-o(1)) \Big )\\
& = M \times (L-o(L)).
\end{align}
\end{proof}

\begin{lemma}\label{lemma4}
For any random variables $X,Y,Z$, such that $Z$ takes values from $\calZ$,
\begin{align}
I(X;Y|Z) - \log(|\calZ|) \le I(X,Y) \le \log(|\calZ|)+I(X;Y|Z). 
\end{align}
\end{lemma}

\begin{proof}
\begin{align}
I(X,Y) & = H(X)-H(X|Y)\\
& \le H(X)- H(X|Y,Z)\\
& \le H(X,Z) - H(X|Y,Z)\\
& = H(Z)+H(X|Z)-H(X|Y,Z)\\
& = H(Z)+I(X,Y|Z)\\
&\le \log(|\calZ|)+I(X,Y|Z) . 
\end{align}
In addition, we write
\begin{align}
I(X,Y) & = H(X)-H(X|Y)\\
& \ge H(X|Z)- H(X|Y)\\
& =   H(X|Z)- H(X|Y,Z) + H(X|Y,Z)- H(X|Y)\\
& = I(X,Y|Z) + H(X|Y,Z)- H(X|Y)\\
& =  I(X,Y|Z) -I(X;Z|Y)\\
& \ge   I(X,Y|Z) -H(Z|Y)\\ 
& \ge   I(X,Y|Z) -H(Z) \\
& \ge   I(X,Y|Z) - \log(|\calZ|). 
\end{align}
\end{proof}

 
 \subsection{Proof of the Converse}

In order to prove the converse, we  show that 
for each $M-$achievable rate $R_M$, we have $R_M\le 1$, meaning that  $D \ge KM$. Let us define\footnote{The numbers
$D_k^{(L)}$, $k\in[1:K]$, are possibly random, due to the  random function assignment. However, the  summation  of them is deterministic, which is equal to $D$.
} $D_k^{(L)} := |\{d\in [1:D]:Q_{d}^{\text{(function)}} = k \}|$ for a given computation scheme. Note that $D=\sum_{k=1}^K \mathbb{E} \Big [ D_k^{(L)} \Big ]$. 
Observe that to prove the converse, it is sufficient to show that $ \mathbb{E} \Big [D_k^{(L)} \Big] \ge M-o(1)$ for each $k$.  Consider a sequence computations schemes, for $L \in \mathbb{N}$,  for the  $M-$achievable rate $R_M$ (see Definition \ref{def3}).   
 To review, note that from  the Fano's inequality, we obtain
 \begin{align}
H(F_{\Sigma_K}^{(L)}F_{\Sigma_{K-1}}^{(L)}\cdots F_{\Sigma_{1}}^{(L)} W_{m}^{(L)}|W^{(L)}_{1:T},Q^{(L)}_{1:D}, A^{(L)}_{1:D},\Sigma)=o(L),
\end{align}
for any $m \in [1:M]$.
For the sake of brevity, we do not use  the superscripts $(.)^{(L)}$ throughout this section anymore. 
Also, let  $
F_{\Sigma}:= F_{\Sigma_K}F_{\Sigma_{K-1}}\cdots F_{\Sigma_1} 
$ and 
$F_{\sim k} :=(F_{1},F_2,\cdots, F_{k-1},F_{k+1},\cdots,F_K)$. 

Fix an integer $k \in [1:K]$. We write
  \begin{align}
  ML & \overset{(a)}{=} H(F_{\Sigma} W_{1:M}|F_{\sim k},\Sigma) 
  \\  &= 
H(F_{\Sigma} W_{1:M}|W_{1:M},Q_{1:D}, A_{1:D},F_{\sim k},\Sigma)+I(F_{\Sigma}  W_{1:M};W_{1:M},Q_{1:D}, A_{1:D}|F_{\sim k},\Sigma) \\
& \le \sum_{m=1}^TH(F_{\Sigma} W_{m}|W_{1:M},Q_{1:D}, A_{1:D},F_{\sim k},\Sigma)+I(F_{\Sigma}  W_{1:M};W_{1:M},Q_{1:D}, A_{1:D}|F_{\sim k},\Sigma)\\& \le \sum_{m=1}^MH(F_{\Sigma} W_{m}|W_{1:M},Q_{1:D}, A_{1:D})+I(F_{\Sigma}  W_{1:M};W_{1:M},Q_{1:D}, A_{1:D}|F_{\sim k},\Sigma)\\
& \overset{(b)}{\le} M\times o(L)+I(F_{\Sigma}  W_{1:M};W_{1:M},Q_{1:D}, A_{1:D}|F_{\sim k},\Sigma)\\
&   \overset{(c)}{\le}o(L)+D \log(K)+I(F_{\Sigma}  W_{1:M};W_{1:M},Q_{1:D}, A_{1:D}|F_{\sim k},\Sigma, Q_{1:D}^{\text{(function)}}),\label{line1}
\end{align}
where $(a)$ follows from Lemma \ref{lemma3}, $(b)$ follows from the correctness property,  and $(c)$ follows from Lemma \ref{lemma4} and the fact that $Q_{1:D}^{(\text{function})}\in [1:K]^D$. 

Let $\delta_{1:D}\in [1:K]^D$ denotes an arbitrary realization of $Q_{1:D}^{(\text{function})}$. Let $\calS = \{ s\in [1:D] : \delta_s = k\}$. Assume $\calS = \{s_1,s_2,\ldots,s_{d_k}\}$ denote the elements of $\calS$, which are ordered increasingly\footnote{
The realization of the random variable $D_k$ is denoted by $d_k$. 
}.

\begin{lemma}\label{lemma5}
For any integer $d \in [1:D]$, 
\begin{align}
I(F_{\Sigma}  W_{1:M};W_{1:M}&,Q_{1:d}, A_{\calS \cup [1:d]}|F_{\sim k},\Sigma, Q_{1:D}^{\text{(function)}}=\delta_{1:D})   \\
&=  I(F_{\Sigma}  W_{1:M};W_{1:M},Q_{1:d-1}, A_{\calS \cup [1:d-1]}|F_{\sim k}, \Sigma, Q_{1:D}^{\text{(function)}}=\delta_{1:D}).
\end{align}
\end{lemma}

\begin{proof}
Conditioning on $F_{\sim k},\Sigma$ and $Q_{1:D}^{\text{(function)}}=\delta_{1:D}$,  we have the following Markov chain:
\begin{align}
F_{\Sigma}  W_{1:M} \to 
(W_{1:M},Q_{1:d-1}, A_{\calS \cup [1:d-1]} ) \to
Q_{d}  \to 
 A_{d} \times  \mathbbm{1}\{d \in \calS \},
\end{align}
which concludes the desired result. 
\end{proof}
 
Using Lemma \ref{lemma5}, we obtain
\begin{align}
&I(F_{\Sigma}  W_{1:M};W_{1:M},Q_{1:D}, A_{1:D}|F_{\sim k}, \Sigma,  Q_{1:D}^{\text{(function)}}=\delta_{1:D})   \\
&= I(F_{\Sigma}  W_{1:M};W_{1:M},Q_{1:D}, A_{\calS \cup [1:D]}|F_{\sim k}, \Sigma, Q_{1:D}^{\text{(function)}}=\delta_{1:D})   \\
&=  I(F_{\Sigma}  W_{1:M};W_{1:M},Q_{1:D-1}, A_{\calS \cup [1:D-1]}|F_{\sim k}, \Sigma, Q_{1:D}^{\text{(function)}}=\delta_{1:D})  \\
&=  I(F_{\Sigma}  W_{1:M};W_{1:M},Q_{1:D-2}, A_{\calS \cup [1:D-2]}|F_{\sim k}, \Sigma, Q_{1:D}^{\text{(function)}}=\delta_{1:D})\\
& =~~~ \cdots \\
 &=  I(F_{\Sigma}  W_{1:M};W_{1:M},Q_{1}, A_{\calS \cup \{1\} }|F_{\sim k}, \Sigma,  Q_{1:D}^{\text{(function)}}=\delta_{1:D}) 
  \\
 &=  I(F_{\Sigma}  W_{1:M};W_{1:M}, A_{\calS }|F_{\sim k}, \Sigma, Q_{1:D}^{\text{(function)}}=\delta_{1:D}) \\
& =   I(F_{\Sigma}  W_{1:M};W_{1:M}|F_{\sim k}, \Sigma,  Q_{1:D}^{\text{(function)}}=\delta_{1:D}) +  I(F_{\Sigma}  W_{1:M}; A_{\calS }|F_{\sim k}, \Sigma,  Q_{1:D}^{\text{(function)}}=\delta_{1:D},W_{1:M}) \\
& \overset{(a)}{\le}   I(F_{\Sigma}  W_{1:M};W_{1:M}|F_{\sim k}, \Sigma) +D \log(K) +  I(F_{\Sigma}  W_{1:M}; A_{\calS }|F_{\sim k}, \Sigma,  Q_{1:D}^{\text{(function)}}=\delta_{1:D},W_{1:M})
\\
& \le  I(F_{\Sigma}  W_{1:M};W_{1:M}|F_{\sim k}, \Sigma) +D \log(K) +  H( A_{\calS }|F_{\sim k}, \Sigma,  Q_{1:D}^{\text{(function)}}=\delta_{1:D},W_{1:M})
\\
& \le  I(F_{\Sigma}  W_{1:M};W_{1:M}|F_{\sim k}, \Sigma) +D \log(K) +  H( A_{\calS })
\\
& \overset{(b)}{\le}  I(F_{\Sigma}  W_{1:M};W_{1:M}|F_{\sim k}, \Sigma) +D \log(K) +  L d_k
\\
& =   H(W_{1:M}|F_{\sim k}, \Sigma)  -  H(F_{\Sigma}  W_{1:M}|F_{\sim k}, \Sigma, W_{1:M})+D \log(K) +  L d_k
\\
& \overset{(c)}{=}  ML  -  M(L-o(L))+D \log(K) +  L d_k,\label{line2}
 \end{align}
where $(a)$ follows from Lemma \ref{lemma4}, $(b)$ follows from the fact that $|\calS| = d_k$, and $(c)$ follows from Lemma \ref{lemma3}.  
Taking the  expectation from (\ref{line2}) and combining with (\ref{line1})  results 
\begin{align}
ML &\le o(L) + D\log(K) + ML - M(L-o(L)) + D\log(K) + L \times  \mathbb{E} \Big [D_k \Big] \\
& = o(L) + 2D \log(K) + L \times  \mathbb{E} \Big [D_k \Big].
\end{align}
Therefore,
\begin{align}
M& \le o(1) + \frac{2D \log(K)}{L} +  \mathbb{E} \Big [D_k \Big]\\
 &\le o(1) + \mathbb{E} \Big [D_k \Big].
\end{align}
Hence, we conclude that $  \mathbb{E} \Big [D_k \Big]  \ge M-o(1)$,   which completes the proof.

\section{Conclusion and Discussion}\label{conclusion}

In this paper, we introduced the problem of  private  function computation. We studied a basic version of  it, namely private sequential function computation, and we showed that even in this special case, the problem is challenging. In particular, we investigated the information theoretic limits of the private sequential function computation problem, and we derived non-trivial lower and upper bounds for its capacity. 

Some of the future directions of this work are as follows: 
\begin{itemize}
\item One direction  is to derive tighter bounds on the capacity of the problem. We conjecture that the scheme proposed in this paper is capacity achieving, and the challenge is to find a tighter converse.

\item 
 An interesting direction is to investigate the capacity for  limited $M$. 

\item Another interesting direction to consider other  classes of functions for the problem of private function computation, rather than only compositions of (large  scale) linear basic functions, where each basic function appears exactly once in the composition. The first step is possibly to consider the case of compositions of the basic functions, such that they can be repeatedly appear in the composition. Another case is when the desired function can be composition of the linear combinations of the the basic functions. 

\item An extension of the proposed problem is the cases where some of servers, up to a given number,  may collude or behave adversarially. 

\end{itemize}

\end{document}